\documentclass[11pt]{article}

\usepackage{amsmath}
\usepackage{amssymb}
\usepackage{latexsym}
\usepackage{graphicx}
\usepackage{epsfig}
\usepackage{stmaryrd}
\usepackage{framed}

\newtheorem{axiom}{Axiom}
\newtheorem{theorem}{Theorem}[section]
\newtheorem{proposition}[theorem]{Proposition}
\newtheorem{definition}[theorem]{Definition}
\newtheorem{lemma}[theorem]{Lemma}

\newtheorem{example}[theorem]{Example}

\newtheorem{corollary}[theorem]{Corollary}
\newtheorem{remark}[theorem]{Remark}

\def\Bbox{
{\unskip\nobreak\hfil\penalty50
\hskip1em\hbox{}\nobreak\hfil{\lower .5pt \hbox{$\Box$}}
\parfillskip=0pt \finalhyphendemerits=0 \par}
}

\def\eop{
\ifmmode {\hbox{\Bbox}} \else \Bbox \fi
}

\newif\ifSuppressEndOfProof\SuppressEndOfProoffalse

\newcommand{\mysim}{\sim\!}

\newcommand{\lord}{\kappa}

\newcommand{\pr}{\mathit{pr}}

\newcommand{\order}{order}

\begin{document}


\title{A Fixed Point Theorem for Non-Monotonic Functions\thanks{This work is
being supported by the Greek General
Secretariat for Research and Technology, the National Development
Agency of Hungary, and the European Commission (European Regional
Development Fund) under a Greek-Hungarian intergovernmental programme of Scientific and Technological collaboration. Project
title: ``Extensions and Applications of Fixed Point Theory for Non-Monotonic Formalisms''.
The first author was also supported by grant no. ANN 110883 from the National Foundation of Hungary
for Scientific Research.}}

\author{Zolt\'{a}n \'{E}sik\\
Department of Computer Science\\
University of Szeged\\
Szeged, Hungary\\
e-mail: \textsf{\vspace{+0.3cm}ze@inf.u-szeged.hu}\\
\and
Panos Rondogiannis\\
Department of Informatics \& Telecommunications\\
University of Athens\\
Athens, Greece\\
e-mail: \textsf{\vspace{+0.3cm}prondo@di.uoa.gr}\\
}

\maketitle
\thispagestyle{empty}
\begin{abstract}
We present a fixed point theorem for a class of (potentially) non-monotonic functions
over specially structured complete lattices. The theorem has as a special case the
Knaster-Tarski fixed point theorem when restricted to the case of monotonic functions
and Kleene's theorem when the functions are additionally continuous. From the practical
side, the theorem has direct applications in the semantics of negation in logic programming.
In particular, it leads to a more direct and elegant proof of the least fixed point
result of~\cite{RW05}. Moreover, the theorem appears to have potential for possible
applications outside the logic programming domain.

\vspace{0.2cm}


\end{abstract}

\date{}

\section{Introduction}
The problem of {\em negation-as-failure}~\cite{AB94,Fit02} in logic programming has received
considerable attention for more than three decades. This research area has proven to be
a quite fruitful one, offering results that range from the very practical to the
very theoretical. On the practical side, negation-as-failure is nowadays used in various
areas of Computer Science (such as in databases, artificial intelligence, and so on).
On the more theoretical side, the study of negation-as-failure has triggered the deeper study
of the nature and repercussions of non-monotonicity in Computer Science. In particular, the study
of the {\em meaning} of logic programs with negation has made evident the necessity of a
fixed point theory for non-monotonic functions.

The fixed point semantics of classical logic programs (ie., programs without negation in the bodies of rules)
was developed by van Emden and Kowalski~\cite{vEK76} and is based on classical fixed point theory
(in particular on the least fixed point theorem of Kleene). However, if negation is introduced
in logic programs, the traditional tools of fixed point theory are no longer applicable due to the
non-monotonicity of the resulting formalism. A crucial step in the study of logic programs with negation was the
introduction of the {\em well-founded semantics}~\cite{vGRS91} which employs a three-valued logic
in order to capture the meaning of these programs. It has been demonstrated that every such program possesses a
minimal three-valued
model which can be constructed as the least fixed point (with respect to the so-called Fitting ordering~\cite{Fit02})
of an appropriate operator associated with the program. The well-founded approach triggered an increased interest in
the study of non-monotonic functions. Such a study aims at developing an abstract fixed point theory
of non-monotonicity which will have diverse applications in various disciplines and research areas.
Results in this direction have been reported in~\cite{DMT00,DMT04,VGD06}. A detailed account of these results
and their relationship with the work developed in this paper, will be given in Section~\ref{related-work-section}.
As a general statement we can say that the existing results indicate that non-monotonic fixed point theory
is a deep area of research that certainly deserves further investigation.

The purpose of the present paper is to develop a novel fixed point theory for an interesting class of
non-monotonic functions. Our motivation comes again from the area of logic programming with negation.
However, our starting point is not one of the traditional constructions of the well-founded semantics
(such as for example~\cite{Prz89,vGel93}). Instead, we start from the {\em infinite-valued semantics}~\cite{RW05}
which is a relatively recent construction that was developed in order to give logical justification to the
well-founded approach. In the infinite-valued semantics the meaning of logic programs with negation is
expressed through the use of an infinite-valued logic. Actually, as it is demonstrated in~\cite{RW05},
every logic program with negation has a unique {\em minimum infinite-valued} model which is
the {\em least fixed point} of an operator with respect to an ordering relation. This minimum
model result extends the well-known minimum model theorem that holds for definite logic programs.
Moreover, the infinite-valued construction is compatible with the well-founded semantics since,
as shown in~\cite{RW05}, if we restrict the minimum infinite-valued model to a three-valued logic,
we get the well-founded model. It is therefore natural to wonder whether the infinite-valued semantics
can form the basis for a novel fixed point theory of non-monotonicity.

In order to develop such a fixed point theory, we keep the essence of the set-theoretic
constructions of~\cite{RW05} but abstract away from all the logic programming related issues.
In particular, instead of studying the set of interpretations of logic programs we consider abstract sets
that possess specific properties. Moreover, instead of focusing on functions from interpretations to
interpretations, we consider functions from abstract sets to abstract sets. More specifically,
our starting point is a complete lattice $(L,\leq)$ equipped with a family of preorderings
indexed by ordinals that give rise to an ordering relation $\sqsubseteq$. We demonstrate that
if the preorderings over $L$ obey certain simple and natural axioms, then the structure $(L,\sqsubseteq)$ is
also a complete lattice. We then prove that a large class of functions $f:L\rightarrow L$ which may not be
monotonic with respect to $\sqsubseteq$, possess a least fixed point with respect to $\sqsubseteq$. Moreover,
we demonstrate that our theorem generalizes both the Knaster-Tarski and the Kleene fixed point theorems (when
$f$ is monotonic or continuous respectively).

The contributions of the present work can be summarized as follows:
\begin{itemize}
\item We develop a fixed point theorem for a class of (potentially) non-monotonic functions
      over specially structured complete lattices. The structure of our lattices stems from
      a simple set of axioms that the corresponding ordering relations have to obey. The proposed
      fixed point theorem appears to be quite general, since, apart from being applicable to a large
      class of non-monotonic functions, it also generalizes well-known fixed point theorems for
      monotonic functions.

\item We demonstrate the versatility of the proposed theorem by deriving a much shorter and
      cleaner proof of the main theorem of~\cite{RW05}. Actually, we demonstrate a much stronger
      result which may be applicable to richer extensions of logic programming.

\item We argue that the proposed theorem may be applicable to other areas apart
      from logic programming. In particular, we demonstrate that the axioms on which the
      fixed point theorem is based, have a variety of other models apart from the set of
      interpretations of logic programs. This fact additionally advocates the generality
      of the proposed approach.
\end{itemize}
The rest of the paper is organized as follows: Section~\ref{infinite-valued-section} introduces the infinite-valued
approach which motivated the present work. This introduction to the infinite-valued approach facilitates
the understanding of the more abstract material of the subsequent sections; moreover, as we are going
to see, the infinite-valued approach will eventually benefit from the abstract setting that will be
developed in the paper. Section~\ref{axioms-section} introduces the partially ordered sets that will be the
objects of our study. Every such set is equipped with an ordering relation whose construction obeys
four simple axioms. The main properties of these sets are investigated. In Section~\ref{lattice-section}
it is demonstrated that every partially ordered set whose ordering relation satisfies the axioms
of Section~\ref{axioms-section}, is a complete lattice. Section~\ref{models-section} presents certain
complete lattices that satisfy the proposed axioms. Section~\ref{fixed-point-section} develops
the novel fixed point theorem for functions defined over the specially structured complete lattices
introduced in the preceding sections. Section~\ref{a-continuous-functions} demonstrates a large class
of functions over which the new fixed point theorem is applicable. As it turns out, the immediate consequence
operator for logic programs with negation falls into this class. In this way we obtain the main result
of~\cite{RW05} as a special case of a much more general theorem. Section~\ref{related-work-section} provides
a comparison with related work and Section~\ref{conclusions-section} concludes the paper with pointers to future work.

In the following, we assume familiarity with the basic notions regarding logic programming (such
as for example~\cite{Llo87}) and of partially ordered sets and particularly lattices
(such as for example~\cite{DP02}).

\section{An Overview of the Infinite-Valued Approach}\label{infinite-valued-section}
In this section we provide the basic notions and definitions behind the infinite-valued approach;
our presentation mostly follows~\cite{RW05}. Some standard technical terminology regarding logic programming
(such as ``atoms'', ``literals'', ``head/body of a rule'', ``ground instance'', ``Herbrand Base'', and so on),
will be used without further explanations (see~\cite{Llo87} for a basic introduction).
\begin{definition}
A (first-order) {\em normal program rule} is a rule with an atom as head and a
conjunction of literals as body. A (first-order) {\em normal logic program} is a
finite set of normal program rules.
\end{definition}

We follow a common practice which dictates that instead of studying finite first-order
logic programs it is more convenient to study their, possibly infinite,
{\em ground instantiations}~\cite{Fit02}:
\begin{definition}
If $P$ is a normal logic program, its associated {\em ground instantiation} $P^*$ is constructed
as follows: first, put in $P^*$ all ground instances of members of $P$; second, if
a rule $A \leftarrow$ with empty body occurs in $P^*$, replace it with $A \leftarrow \mbox{\tt true}$;
finally, if the ground atom $A$ is not the head of any member of $P^*$, add
$A \leftarrow \mbox{\tt false}$.
\end{definition}

Notice that, by construction, $P^*$ is a propositional program that has a possibly infinite
but countable number of rules (since the Herbrand Base of a normal logic program is countable).
To simplify our presentation, in the rest of the paper we will not talk explicitly about ``the
ground instantiation $P^*$ of a program $P$''. Instead, we will assume that we study programs
that are propositional and have a countable number of rules. Since we will be dealing
with propositional programs, we will often talk about ``propositional atoms'' and ``propositional
literals'' that appear in the rules of our programs.

The basic idea behind the infinite-valued approach is that in order to obtain a
minimum model semantics for logic programs with negation, it is necessary to consider
a refined multiple-valued logic which will allow the meaning of negation-as-failure to
be expressed properly. Let $\Omega$ denote the first uncountable ordinal.
Then, the logic of~\cite{RW05} contains one $F_{\alpha}$ and one $T_{\alpha}$ for each
countable ordinal $\alpha$ (ie., for all $\alpha<\Omega$), and also an intermediate truth value
denoted by 0. The ordering of the truth values is as follows:
$$F_0 < F_1 < \cdots < F_{\omega} < \cdots < F_{\alpha} < \cdots < 0 <
\cdots < T_\alpha < \cdots <T_\omega < \cdots <T_1 <T_0$$
Intuitively, $F_0$ and $T_0$ are the classical {\em False} and {\em True} values and
0 is the {\em undefined} value. The intuition behind the new values is that they express
different levels of truthfulness and falsity. Alternatively, the ordinal indices of truth
values can be shown to correspond to the level at which truth or falsity of atoms
is derived during the well-founded construction (see~\cite{RW05}[Proof outline of Theorem 7.6]).
In the following we denote by $V$ the set consisting of the above truth values.
\begin{definition}
The {\em order} of a truth value is defined as: $order(T_{\alpha}) = \alpha$,
$order(F_{\alpha}) = \alpha$ and $order(0) = +\infty$.
\end{definition}

Interpretations of programs are defined as follows:
\begin{definition}
An (infinite-valued) interpretation $I$ of a program $P$ is a function from the set of
propositional atoms of $P$ to $V$.
\end{definition}

Interpretations can be extended to apply to literals,
to conjunctions of literals and to the two constants {\tt true} and {\tt false}.
Of special interest is the way that negation is treated (and which intuitively
expresses the fact that the more times negation is iterated, the more it approaches
to the intermediate value 0).
\begin{definition}\label{interpretation}
Let $I$ be an interpretation of a program $P$. Then, $I$ can be extended
as follows:
\begin{itemize}
\item For every negative atom $\mysim p$ appearing in $P$:
      \[
             I(\mysim p) = \left\{
                             \begin{array}{ll}
                             T_{\alpha + 1} & \mbox{if $I(p) = F_\alpha$}\\
                             F_{\alpha + 1} & \mbox{if $I(p) = T_\alpha$}\\
                             0              & \mbox{if $I(p) = 0$}
                             \end{array}
                      \right.
      \]

      \vspace{-0.5cm}
\item For every conjunction of literals $l_1,\ldots,l_n$ appearing as
      the body of a rule in $P$:
      \[
       I(l_1,\ldots,l_n) = min\{I(l_1),\ldots,I(l_n)\}
      \]
\end{itemize}
Moreover, $I(\mbox{\tt true}) = T_0$ and $I(\mbox{\tt false}) = F_0$.
\end{definition}

The notion of satisfiability of a rule can be defined as follows:
\begin{definition}
Let $P$ be a program and $I$ an interpretation of $P$. Then, $I$ {\em satisfies}
a rule $p \leftarrow l_1,\ldots,l_n$ of $P$ if $I(p) \geq I(l_1,\ldots,l_n)$. Moreover,
$I$ is a {\em model} of $P$ if $I$ satisfies all rules of $P$.
\end{definition}

As it is customary in the theory of logic programming, it would be desirable if
we could prove that every program has a unique minimum model with respect to an
ordering relation. The first idea that comes to mind is to use the ``pointwise''
ordering of interpretations, namely $I\leq J$ iff for all propositional symbols $p$
it holds $I(p) \leq J(p)$. Actually, this is the ordering used in classical logic
programming in order to establish the minimum model property (see for example~\cite{Llo87}),
where the set of truth values is $\{\mathit{False},\mathit{True}\}$ with
$\mathit{False} < \mathit{True}$. As it can easily be checked, the
pointwise ordering of interpretations does not lead to a minimum model
result in the case of logic programs with negation. However, the pointwise
relation $\leq$ (more precisely, its abstract analogue) will assist us
in developing an abstract fixed point theorem for non-monotonic functions.

In order to obtain the minimum model property for logic programs with negation,
a more refined ordering on interpretations was defined in~\cite{RW05}.
This ordering was actually motivated by the fact that the construction of the
well-founded model~\cite{Prz89} proceeds in stages and therefore the atoms
that belong to the lower stages must (intuitively) be given a ``higher priority''
than those belonging to upper stages. We will need the following definition:
\begin{definition}
Let $P$ be a program, $I$ an interpretation of $P$ and $v\in V$.
Let $B_P$ be the set of all propositional symbols that appear in program $P$. Then
$I \parallel v = \{p \in B_P \mid I(p) = v \}$.
\end{definition}

We can now proceed to the definition of the ordering relations that are needed in order
to obtain the minimum model property:
\begin{definition}
Let $I$ and $J$ be interpretations of a given program $P$ and $\alpha$ be a countable
ordinal. We write $I=_{\alpha} J$, if for all $\beta \leq a$,
$I\parallel T_{\beta} = J\parallel T_{\beta}$
and $I\parallel F_{\beta} = J\parallel F_{\beta}$.
\end{definition}
\begin{definition}
Let $I$ and $J$ be interpretations of a given program $P$ and $\alpha$ be a countable ordinal. We write
$I\sqsubset_{\alpha} J$, if for all $\beta < a$, $I=_{\beta} J$ and
either  $I\parallel T_{\alpha} \subset J\parallel T_{\alpha}$ and $I\parallel F_{\alpha} \supseteq J\parallel F_{\alpha}$,
or $I\parallel T_{\alpha} \subseteq J\parallel T_{\alpha}$ and $I\parallel F_{\alpha} \supset J\parallel F_{\alpha}$.
We write $I \sqsubseteq_{\alpha} J$ if $I =_{\alpha} J$ or $I \sqsubset_{\alpha} J$.
\end{definition}
\begin{definition}
Let $I$ and $J$ be interpretations of a given program $P$. We write
$I\sqsubset J$, if there exists a countable ordinal $\alpha$ such that
$I \sqsubset_{\alpha} J$. We write $I\sqsubseteq J$ if either
$I = J$ or $I\sqsubset J$.
\end{definition}

It is easy to see that the relation $\sqsubseteq$ on the set of interpretations
of a given program, is a partial order (ie. it is reflexive, transitive and antisymmetric);
actually, as we will see in Section~\ref{lattice-section}, the set of interpretations of
a program equipped with the relation $\sqsubseteq$ forms a complete lattice.
On the other hand, for every countable ordinal $\alpha$, the relation $\sqsubseteq_{\alpha}$
is a preorder (ie. reflexive and transitive).

In order to study the fixed point semantics of the programs we consider,
we can easily define an immediate consequence operator $T_P$:
$$T_P(I)(p) = \bigvee\{I(l_1,\ldots,l_n) \mid (p \leftarrow l_1,\ldots,l_n)\in P\}$$
where $\bigvee$ is the obvious least upper bound operation on the set $V$.
Notice that, as discussed in~\cite{RW05}[Example 5.7], $T_P$ is not
in general monotonic with respect to $\sqsubseteq$. This fact makes the proof
of the following theorem nontrivial.

\begin{theorem}\cite{RW05}[Corollary 7.5, page 460]\label{least}
For every program $P$, $T_P$ has a least fixed point $M_P$ (with respect to $\sqsubseteq$).
\end{theorem}

Actually, it is easy to show (see~\cite{RW05}) that $M_P$ is a model of $P$ (in fact, the least model of $P$ with
respect to $\sqsubseteq$), and therefore it can be taken as the intended meaning of $P$. Moreover,
$M_P$ is directly connected to the well-founded model of $P$, as the following theorem
from~\cite{RW05} suggests:
\begin{theorem}\cite{RW05}[Theorem 7.6, page 460]
Let $N_P$ be the interpretation that results from $M_P$ by collapsing
all true values to $\mathit{True}$ and all false values to $\mathit{False}$.
Then, $N_P$ is the well-founded model of $P$.
\end{theorem}
\begin{example}\label{example-program}
Consider the program:
\[
\begin{array}{lll}
 {\tt p} & \leftarrow & \mysim {\tt q}\\
 {\tt q} & \leftarrow & \mysim {\tt r}\\
 {\tt s} & \leftarrow & {\tt p}\\
 {\tt s} & \leftarrow & \mysim {\tt s} \\
 {\tt r} & \leftarrow & \mbox{\tt false}
\end{array}
\]
One can easily compute the infinite-valued model of the program (see~\cite{RW05}[pages 454-455] for
details), which is equal to:
$$M_P = \{({\tt p},F_2),({\tt q},T_1),({\tt r},F_0),({\tt s},0)\}$$
Moreover, the well-founded model of the program is
$$N_P = \{({\tt p},\textit{False}),({\tt q},\textit{True}),({\tt r},\textit{False}),({\tt s},0)\}$$
and has been obtained from $M_P$ in the obvious way.
\end{example}

The proof of Theorem~\ref{least} was performed in~\cite{RW05} using techniques that were
specifically tailored to the case of logic programs. In the next sections we abstract away from
the issues regarding logic programming, and we obtain a general and abstract fixed point theorem
from which the above result follows in a very direct way.

\section{Axioms and their Consequences}\label{axioms-section}
In this section we keep the essence of the set-theoretic constructions presented in the
previous section but we abstract away from the logic programming related issues.
Our starting point is not anymore ``the set of interpretations'' but instead a complete
lattice $(L,\leq)$ equipped with a family of preorderings indexed by ordinals that give rise to an
ordering relation $\sqsubseteq$. Recall that in Section~\ref{infinite-valued-section} the relation $\leq$ corresponded
to the pointwise ordering of interpretations. However, in our new abstract setting our only
initial assumption is that $(L,\leq)$ is a complete lattice, ie., we have no specific initial knowledge
regarding the nature of the ordering $\leq$. Similarly, the relations $\sqsubseteq_{\alpha}$, $\sqsubset_{\alpha}$ and
so on, are not anymore relations over interpretations but over the elements of the
abstract set $L$ that obey specific axioms. Notice also that in our new setting the ordinals
$\alpha$ need not be countable; it suffices to assume that every $\alpha$ is less than a fixed ordinal
$\lord>0$.

The rest of this section is structured as follows. In Subsection~\ref{axioms-subsection} we
present four axioms that the relations over $L$ have to obey (in order for the fixed point
theorem that will be developed in Section~\ref{fixed-point-section} to be applicable). Subsection~\ref{some-consequences-subsection}
presents some easy consequences of the axioms. Subsection~\ref{slices} demonstrates that if $L$ obeys the
four axioms then every element $x$ of $L$ can be represented as the least upper bound of a set of ``simpler'' elements
of $L$ called the {\em slices} of $x$. Finally, Subsection~\ref{compatible-sequences} demonstrates that
one can also view the elements of $L$ as the least upper bounds of a special kind of sequences termed {\em compatible
sequences}. Recapitulating, the material in this section gives a characterization of the structure
of the sets that will be the objects of our investigation.

\subsection{The Axioms}\label{axioms-subsection}
Suppose that $(L,\leq)$ is a complete lattice in which the least upper bound
operation is denoted by $\bigvee$ and the least element is denoted by $\perp$.
Let $\lord >0$ be a fixed ordinal. We assume that for each ordinal $\alpha<\lord$, there exists
a preordering $\sqsubseteq_\alpha$ on $L$. We denote with $ =_\alpha$ the equivalence relation determined by $\sqsubseteq_\alpha$.
We define $x \sqsubset_\alpha y$ iff $x \sqsubseteq_\alpha y$ but $x =_\alpha y$
does not hold. Finally, we define $\sqsubset  =   \bigcup_{\alpha < \lord}  \sqsubset_\alpha$
and let $x \sqsubseteq y$ iff $x \sqsubset y$ or $x = y$.

Intuitively, one can think of elements of $L$ as consisting of ``components'',
one for every $\alpha < \lord$. Roughly speaking, the relations $\leq$ and $\sqsubseteq$
correspond to two different ways of comparing elements of $L$: the relation $\leq$ corresponds
to a ``conventional'' (for example ``pointwise'') way of comparing the elements of $L$; the
relation $\sqsubseteq$ corresponds to a ``lexicographic'' comparison, in which we start comparing
the 0-th level components, then the first level components, and so on, until we reach a decision.
When we write $x=_\alpha y$ we mean that $x$ and $y$ are equal for all components up to level $\alpha$.
When we write $x\sqsubseteq_\alpha y$ we mean that $x$ is equal to $y$ for all $\beta < \alpha$ and
it is either equal or smaller than $y$ at level $\alpha$. Finally, $x\sqsubseteq y$ means that
either $x=y$ or there exists some $\alpha$ such that $x$ and $y$ are equal in all components
less than $\alpha$ and $x$ is genuinely smaller than $y$ in the $\alpha$-th component. The axioms
that will be given shortly express these intuitions (and certain additional properties
of our ordering relations).  Notice that the ``components'' of an element of $L$ do not appear
explicitly in the axioms, but as we will see in Subsection~\ref{slices}, their existence is
implied by the axioms. Notice also that when reading the axioms it is useful to think
of the relations on interpretations introduced in Section~\ref{infinite-valued-section}.
Actually, at the end of the present subsection we will see that these relations satisfy the
proposed axioms and therefore they provide a model of our axioms. As we will see in
Section~\ref{models-section} other models of the axioms also exist.

The first two axioms state restrictions regarding the relations $\sqsubseteq_\alpha$ and $=_\alpha$:
\begin{framed}
\begin{axiom}\label{axiom1}
For all ordinals $\alpha < \beta < \lord$, $\sqsubseteq_\beta$ is included in $=_\alpha$.
\end{axiom}
\end{framed}
\begin{remark}
\label{remark0}
From the above axiom it follows that for all $\alpha < \beta$ the relation $=_\beta$ is
included in the relation $=_\alpha$.
\end{remark}
\begin{framed}
\begin{axiom}\label{axiom2}
$\bigcap_{\alpha<\lord} =_\alpha$ is the identity relation on $L$.
\end{axiom}
\end{framed}
Given an ordinal $\alpha < \lord$ and $x \in L$, define
\begin{eqnarray*}
(x]_\alpha &=& \{ y \in L : \forall \beta < \alpha\quad   x =_\beta y\}.
\end{eqnarray*}
Note that $x \in (x]_\alpha$ holds. Notice also that for every $x\in L$, $(x]_0=L$.
In the following axiom and also in the rest of the paper, $X \sqsubseteq_\alpha y$ should
be interpreted in the standard way, namely as ``for every $x\in X, x \sqsubseteq_\alpha y$''.
\begin{framed}
\begin{axiom}\label{axiom3}
For each $x\in L$, for every ordinal $\alpha < \lord$, and for any $X \subseteq (x]_\alpha$
there is some $y \in (x]_\alpha$ such that:
\begin{itemize}
\item $X \sqsubseteq_\alpha y$, and
\item for all $z \in (x]_\alpha$, if $X \sqsubseteq_\alpha z$ then
      $y \sqsubseteq_\alpha z$ and $y \leq z$.
\end{itemize}
\end{axiom}
\end{framed}
The element $y$ specified by the above axiom is easily shown to be unique and we denote
it by $\bigsqcup_\alpha X$.

In the case where $X$ is a \emph{nonempty} set, we can use the notation
$\bigsqcup_\alpha X$ without specifying explicitly a particular element $x$
such that $X \subseteq (x]_\alpha$. This can be done because if
$X \neq \emptyset$ and for some $x,y\in L$ it holds $X \subseteq (x]_\alpha$ and
$X \subseteq (y]_\alpha$, then $(x]_\alpha = (y]_\alpha$, since if
$z \in L$, then we have $x =_\beta z =_\beta y$ for all $\beta < \alpha$.
Alternatively, we may freely use the notation $\bigsqcup_\alpha X$ for all
nonempty $X$ such that $x=_\beta y$ holds for all $x,y\in X$.

\begin{remark}
\label{rem-remark1}
When $X=\emptyset$ and we consider $X$ to be a subset of $(x]_\alpha$, then
$\bigsqcup_\alpha X$ is both the $\leq$-least and a $\sqsubseteq_\alpha$-least
element of $(x]_\alpha$. Thus, for every $x\in L$ and ordinal $\alpha < \lord$,
$(x]_\alpha$ has a $\leq$-least element, which is, at the same time,
a $\sqsubseteq_\alpha$-least element. Thus, when $X$ is empty and $x = \perp$,
we obtain that $\perp \sqsubseteq_\alpha y$ for all $y \in (\perp]_\alpha$.
Notice also that since $(\perp]_0=L$, it holds $\perp \sqsubseteq_0 y$ for all
$y \in L$.
\end{remark}
\begin{framed}
\begin{axiom}\label{axiom4}
For every non-empty $X\subseteq L$ and ordinal $\alpha < \lord$, if $y =_\alpha x$ for all $x \in X$,
then $y =_\alpha (\bigvee X)$.
\end{axiom}
\end{framed}

In the rest of the paper we will assume, without loss of generality, that $\lord$
is a limit ordinal. To see that this is not a restriction, assume that $\lord$ is
a successor ordinal, say $\gamma+1$, and that the corresponding relations satisfy
the above axioms for all ordinals less than $\lord$. By Axioms~\ref{axiom1} and~\ref{axiom2}
it follows that $=_\gamma$ should be the equality relation on $L$. But then, we could
add relations $\sqsubseteq_\beta$, each one of them being equal to the equality relation on $L$,
for all $\gamma < \beta <\overline{\lord}$, where $\overline{\lord}$ is the least limit ordinal
greater than $\lord$. It is easy to check that the above axioms now also hold for all ordinals
less than $\overline{\lord}$. Moreover, the relation $\sqsubset$ remains the same if we replace $\lord$
with $\overline{\lord}$ because $\bigcup_{\alpha < \lord}  \sqsubset_\alpha = \bigcup_{\alpha < \overline{\lord}}  \sqsubset_\alpha$. Finally, if a function $f:L\rightarrow L$ preserves all the relations
$\sqsubseteq_\alpha$ for all $\alpha < \lord$ then $f$ also preserves all the relations
$\sqsubseteq_\alpha$ for all $\alpha < \overline{\lord}$ (notice that this is a key assumption
of Theorem~\ref{main-fixed-point-theorem}). As a result, the two main theorems of the paper (Theorem~\ref{complete-lattice-theorem} and Theorem~\ref{main-fixed-point-theorem}) that involve
the relations $\sqsubseteq_\alpha$ and the relation $\sqsubseteq$, are not affected if we replace
the successor ordinal $\lord$ with the limit ordinal $\overline{\lord}$ as described above. Therefore,
in the rest of the paper instead of distinguishing cases regarding whether $\lord$ is successor or limit,
we can safely consider that $\lord$ is always limit and proceed with this assumption.

In the rest of the paper, we will often talk about ``{\em models of the Axioms 1-4}''
(or simply ``{\em models}''). More formally:
\begin{definition}
A {\em model of Axioms 1-4} or simply {\em model} consists of a complete lattice $(L,\leq)$,
an ordinal $\lord >0$ and a set of preorders $\sqsubseteq_\alpha$ for every $\alpha < \lord$,
such that Axioms 1-4 are satisfied.
\end{definition}

As we are going to see in Section~\ref{models-section}, there exist various models
of the axioms. However we already know one such model. More specifically, it is
not hard to check that the set of infinite-valued interpretations together with
the relations $\sqsubseteq_\alpha$, $\alpha < \Omega$, introduced in
Section~\ref{infinite-valued-section}, form a model of the Axioms 1-4. More specifically,
let $Z$ be a non-empty set (corresponding, for example, to the propositional variables of a program).
Consider the set $V^Z$ where $V$ is the set of truth values introduced in Section~\ref{infinite-valued-section}.
Given $I,J \in V^Z$, we write $I \leq J$ iff for all $z \in Z$, $I(z) \leq J(z)$. Notice that $(V^Z,\leq)$ is a
complete lattice where for any $X \subseteq V^Z$, $(\bigvee X)(z) = \bigvee\{I(z): I \in X\}$.
Consider now the relations $\sqsubseteq_\alpha$ and $=_\alpha$ as defined in
Section~\ref{infinite-valued-section}.

It is straightforward to check that Axioms 1 and 2 are satisfied. In order to show that
Axiom 3 is satisfied, let $\alpha < \kappa = \Omega$, let $I \in V^Z$ and $X\subseteq V^Z$
such that $X \subseteq (I]_\alpha$.  We define:
\[
\begin{array}{ll}
(\bigsqcup_{\alpha} X)(z) = \left\{\begin{array}{ll}
                                    I(z),       &\mbox{if $\textit{order}(I(z))<\alpha$}\\
                                    T_{\alpha}, &\mbox{if there exists $J\in X$ such that $J(z)=T_\alpha$}\\
                                    F_{\alpha}, &\mbox{if for all $J\in X$,  $J(z)=F_\alpha$}\\
                                    F_{\alpha+1}, &\mbox{otherwise}
                                    \end{array}
                                  \right.
\end{array}
\]
Notice that when $X$ is empty, then $(\bigsqcup_\alpha X)(z) = F_\alpha$.

It is easy to verify that $\bigsqcup_{\alpha} X$ satisfies the requirements of Axiom 3.

Finally, for Axiom 4, let $X \subseteq V^Z$ be a nonempty set and assume that $J=_\alpha I$, for all $I\in X$.
We show that $\bigvee X =_\alpha J$. Consider $z\in Z$ such that $\textit{order}(J(z))\leq \alpha$. Then,
$J(z)=I(z)$, for all $I\in X$. Therefore $(\bigvee X)(z) = \bigvee\{I(z):I\in X\} = J(z)$.
Conversely, consider $z \in Z$ such that $\textit{order}((\bigvee X)(z))\leq \alpha$. This implies that
$\textit{order}(\bigvee\{I(z):I\in X\}) \leq \alpha$, ie., $\textit{order}(I(z))\leq \alpha$,
for all $I \in X$. Therefore, $I(z)=J(z)$ for all $I\in X$ and $(\bigvee X)(z) = \bigvee\{I(z):I\in X\} = J(z)$.

This completes the proof that the set of interpretations forms a model of the axioms.
In the rest of this paper, this particular model will be referred as
``{\em the standard model}'' (since this has been the motivating factor
behind the present work).

\subsection{Some Consequences of the Axioms}\label{some-consequences-subsection}
Based on the above axioms, various properties can easily be established regarding the
aforementioned relations.
\begin{lemma}\label{some-consequences}
Let $\alpha,\beta < \lord$. The following properties hold:
\begin{enumerate}

\item[(a)] $\bigcap_{\alpha<\lord} \sqsubseteq_\alpha$ is the equality relation on $L$.

\item[(b)] the relation $\sqsubset_\alpha \circ \sqsubset_\beta$ is included in the relation
           $\sqsubset_{\min\{\alpha,\beta\}}$.

\item[(c)] if $\alpha < \beta$ then $\sqsubset_\alpha$ and $\sqsubset_\beta$ are disjoint.

%

\item[(d)] if $\alpha < \beta$ then $=_\alpha \circ \sqsubseteq_\beta \quad = \quad \sqsubseteq_\beta \circ =_\alpha
           \quad = \quad =_\alpha$.
\end{enumerate}
\end{lemma}
\begin{proof}

To prove (a),
consider $(x,y) \in \bigcap_{\alpha<\lord} \sqsubseteq_\alpha$.
Then, for all $\alpha < \lord$, $x \sqsubseteq_\alpha y$ which, by Axiom 1, implies that for all $\beta < \lord$, $x=_\beta y$. By Axiom 2, $x=y$. On the other hand if $x = y$ then $x =_\alpha y$ and $x \sqsubseteq_\alpha y$
for all $\alpha < \kappa$.

Consider now statement (b). If $\alpha=\beta$ then the statement obviously holds.
Assume therefore that $\alpha < \beta$ (the case $\beta < \alpha$ is symmetric).
Let $(x,z) \in  \sqsubset_\alpha \circ \sqsubset_\beta$. Then, there exists some
$y$ such that $x \sqsubset_\alpha y$ and $y \sqsubset_\beta z$; therefore,
$x \sqsubseteq_\alpha y$ and $y \sqsubseteq_\beta z$. By Axiom 1 and the fact that
$y \sqsubseteq_\beta z$ we get $y \sqsubseteq_\alpha z$ and therefore, by the transitivity
of $\sqsubseteq_\alpha$, $x \sqsubseteq_\alpha z$. But then, either $x =_\alpha z$ or
$x \sqsubset_\alpha z$. Assume for the sake of contradiction that $x =_\alpha z$.
Then, $z \sqsubseteq_\alpha x$ and since $y \sqsubseteq_\alpha z$ we get $y \sqsubseteq_\alpha x$.
But since $x \sqsubseteq_\alpha y$, we get $x=_\alpha y$, which is a contradiction
because $x \sqsubset_\alpha y$. Therefore, $x \sqsubset_\alpha z$, ie. $x \sqsubset_{\min\{\alpha,\beta\}} z$.

To establish statement (c), assume that $\sqsubset_\alpha$ and $\sqsubset_\beta$ are not disjoint. Then
there exist $x,y$ such that $x \sqsubset_\alpha y$ and $x \sqsubset_\beta y$. But $x \sqsubset_\beta y$ implies
$x \sqsubseteq_\beta y$ and by Axiom 1 we get $x =_\alpha y$ (contradiction).

Finally, to prove statement (d), assume that $(x,z)$ belongs to the relation $=_\alpha \circ \sqsubseteq_\beta$.
Then, there exists $y$ such that $x =_\alpha y$ and $y \sqsubseteq_\beta z$. By Axiom 1 we get that
$y=_\alpha z$ and therefore $x =_\alpha z$. Conversely, assume that $x=_\alpha z$. Then, since
$x=_\alpha x$ and $x=_\alpha z$  we get that $(x,z)$ belongs to the relation $=_\alpha \circ \sqsubseteq_\beta$.
A similar proof shows that the relations $\sqsubseteq_\beta \circ =_\alpha$ and  $=_\alpha$ are equal.
\end{proof}

\begin{lemma}\label{partial-order}
The relation $\sqsubseteq$ is a partial order.
\end{lemma}
\begin{proof}
Since by Lemma~\ref{some-consequences}(b) $\sqsubset_\alpha \circ \sqsubset_\beta$ is included
in $\sqsubset_{\min\{\alpha,\beta\}}$, the relation $\sqsubset$
is transitive as is the relation $\sqsubseteq$. It is clear
that $\sqsubseteq$ is reflexive. If $x \sqsubseteq y$ and $y \sqsubseteq x$,
and $x \neq y$, then there exist ordinals $\alpha,\beta$ with
$x \sqsubset_\alpha y$ and $y \sqsubset_\beta x$. Let $\gamma = \min\{\alpha,\beta\}$.
Then by Lemma~\ref{some-consequences}(b) $x \sqsubset_\gamma x$, which is a contradiction
since $x =_\gamma x$. Thus, if $x \sqsubseteq y$ and $y \sqsubseteq x$, then $x = y$.
\end{proof}
\begin{lemma}\label{inclusion-lemma}
The relation $\sqsubseteq$ is included in $\sqsubseteq_0$.
\end{lemma}
\begin{proof}
If $x \sqsubseteq y$, then either $x = y$ or $x \sqsubset_\alpha y$ for some $\alpha$.
If $x = y$ then by reflexivity we have $x =_0 y$ and therefore $x\sqsubseteq_0 y$.
If $x \sqsubset_0 y$, then $x \sqsubseteq_0 y$ clearly holds. If $x \sqsubset_\alpha y$
for some $\alpha > 0$, then, by Axiom~\ref{axiom1}, $x =_0 y$, so that $x \sqsubseteq_0 y$ again.
\end{proof}


For each $x\in L$ and ordinal $\alpha < \lord$, let $[x]_\alpha = \{y \in L: x =_\alpha y\} =
\{y \in L: \forall \beta \leq \alpha\ x =_\beta y\}$ (where the second equality is
justified by Remark~\ref{remark0}). Note that $[x]_\alpha \subseteq (x]_\alpha$
and $[x]_\alpha = (x]_{\alpha + 1}$.

\begin{lemma}
\label{lem-lem0}
Let $\alpha < \lord$. Suppose that $y = \bigsqcup_\alpha X$,
where $X\subseteq (x]_\alpha$ for some $x\in L$.
Then $y$ is the $\leq$-least element of $[y]_\alpha$.
Also, $y$ is a $\sqsubseteq_{\alpha + 1}$-least
element of $[y]_\alpha$, ie., $y \sqsubseteq_{\alpha + 1} z$
for all $z \in [y]_\alpha$.
\end{lemma}

\begin{proof}
Suppose $z \in [y]_\alpha$.
Then since $y =_\alpha z$ and $y \in (x]_\alpha$, by Remark~\ref{remark0}
we have $z =_\beta x$ for all $\beta < \alpha$, so that $z \in (x]_\alpha$.
Since $X \sqsubseteq_\alpha y$ and $y =_\alpha z$,
also $X \sqsubseteq_\alpha z$. 
Thus, $y \leq z$, by the definition of $y$.

Since $[y]_\alpha = (y]_{\alpha + 1}$, and since $y$
is $\leq$-least in $[y]_\alpha$, $y$ is $\leq$-least in $(y]_{\alpha + 1}$.
But by Remark~\ref{rem-remark1}, the $\leq$-least element of $(y]_{\alpha + 1}$ is
$\bigsqcup_{\alpha + 1}\emptyset$, where $\emptyset$ is viewed as
a subset of $(y]_{\alpha+1}$, which is also a $\sqsubseteq_{\alpha + 1}$-least
element of $(y]_{\alpha + 1}$. Thus,  $y \sqsubseteq_{\alpha + 1} z$
for all $z \in [y]_\alpha$.
\end{proof}

\begin{lemma}
\label{lem-lem1}
Let $x \in L$ and $\alpha < \lord$ be an ordinal.
Suppose that $X \subseteq (x]_\alpha$ and $z\in (x]_\alpha$ with $X \sqsubseteq z$.
Then $\bigsqcup_\alpha X \sqsubseteq_\alpha z$ and $\bigsqcup_\alpha X \leq z$.
\end{lemma}

\begin{proof}
Since $X \sqsubseteq z$, for each $x \in X$
either $x = z$ or there is some ordinal $\beta$ with $x\sqsubset_\beta z$.
Since $z \in (x]_\alpha$, we must have $\beta \geq \alpha$.
In either case, $x \sqsubseteq_\alpha z$, so that $X \sqsubseteq_\alpha z$.
Thus, by the definition of $\bigsqcup_\alpha X$, we have that
$\bigsqcup_\alpha X \sqsubseteq_\alpha z$ and $\bigsqcup_\alpha X \leq z$.
\end{proof}

\subsection{Slices}\label{slices}
In this subsection we demonstrate that each element $x$ of $L$ can
be created as the least upper bound of a set of ``simpler'' elements
of $L$ called the {\em slices} of $x$. In the following we establish
several properties of slices.
\begin{lemma}
\label{lem-slice}
For every $x \in L$ and ordinal $\alpha < \lord$,
 $y = \bigsqcup_\alpha \{ x \}$ is both the $\leq$-least and a
$\sqsubseteq_{\alpha+1}$-least element of $[x]_\alpha$. In particular,
$y \in [x]_\alpha$.
\end{lemma}

\begin{proof}
By Lemma~\ref{lem-lem0}, $y$ is both the $\leq$-least element and a
$\sqsubseteq_{\alpha + 1}$-least element of $[y]_\alpha$.
If we can show that $y =_\alpha x$, then it follows that
$[y]_\alpha = [x]_\alpha$, and thus $y$ is both the $\leq$-least and a
$\sqsubseteq_{\alpha+1}$-least element of $[x]_\alpha$.

By definition, we have that $x \sqsubseteq_\alpha  \bigsqcup_\alpha \{x\}  = y$.
Moreover, by Axiom~\ref{axiom3}, for all $z$, if $x \sqsubseteq_\alpha z$, then
$y \sqsubseteq_\alpha z$. Therefore, since $x \sqsubseteq_\alpha x$, it follows that
$y \sqsubseteq_\alpha x$. Since both $x \sqsubseteq_\alpha y$ and
$y \sqsubseteq_\alpha x$ hold, $x =_\alpha y$.
\end{proof}


\begin{definition}
For every $x \in L$ and $\alpha < \lord$, we denote $\bigsqcup_\alpha \{x\}$
by $x|_\alpha$ and call it the \emph{slice of $x$ at $\alpha$}.
\end{definition}
\begin{remark}
Consider the standard model discussed at the end of Subsection~\ref{axioms-subsection}
and let $I$ be an infinite-valued interpretation. Then:
\[
\begin{array}{ll}
(I|_\alpha)(z) = (\bigsqcup_{\alpha} \{I\})(z) = \left\{\begin{array}{ll}
                                                       I(z),         &\mbox{if $\textit{order}(I(z))\leq\alpha$}\\
                                                       F_{\alpha+1}, &\mbox{otherwise}
                                                       \end{array}
                                                 \right.
\end{array}
\]
In other words, the slice of an interpretation $I$ at level $\alpha$ is an interpretation identical
to $I$ for all levels less than or equal to $\alpha$, which however assigns a ``default value'' equal
to $F_{\alpha+1}$ to all variables that in $I$ possess a value of order greater than $\alpha$.
\end{remark}

By Lemma~\ref{lem-slice}, we have that $x =_\alpha x|_\alpha$,
and $x|_\alpha$ is both the $\leq$-least and a
$\sqsubseteq_{\alpha + 1}$-least element of $L$ which
is $=_\alpha$-equivalent to $x$.
In particular, $\perp\kern-.4em|_\alpha = \perp$ for all $\alpha$, since $\perp$ is the
$\leq$-least element of $[\perp]_\alpha$.

\begin{lemma}
\label{lem-slice added}
The following conditions are equivalent for all $x \in L$:
\begin{itemize}
\item[(a)] $x = x|_\alpha$.
\item[(b)] $x = \bigsqcup_\alpha \{x\}$.
\item[(c)] $x = \bigsqcup_\alpha\{y\}$ for some $y\in L$.
\item[(d)] $x = \bigsqcup_\alpha X$ for some $X \subseteq (z]_\alpha$, where $z\in L$.
\end{itemize}
\end{lemma}

\begin{proof}
The first two conditions are equivalent by definition. It is clear that
condition (b) implies (c) which in turn implies (d). Suppose now that
(d) holds. It is clear that $\{ x \}\sqsubseteq_\alpha x$. Let $y\in L$
with $\{x\} \sqsubseteq_\alpha y$. Then also $X \sqsubseteq_\alpha y$
and since $x = \bigsqcup_\alpha X$ it holds $x \leq y$. We conclude that
$x =\bigsqcup_\alpha \{ x \}$.
\end{proof}

\begin{lemma}
\label{lem-alphaeq}
For all $x,y \in L$ and ordinal $\alpha < \lord$,
$x =_\alpha y$ iff $x|_\alpha =_\alpha  y|_\alpha$ iff $x|_\alpha = y|_\alpha$.
Similarly, $x \sqsubset_\alpha y$ iff $x|_\alpha \sqsubset_\alpha y|_\alpha$
and $x \sqsubseteq_\alpha y$ iff $x|_\alpha \sqsubseteq_\alpha y|_\alpha$.
\end{lemma}

\begin{proof}
We have $x =_\alpha x|_\alpha$ and $y =_\alpha y|_\alpha$. Thus,
$x =_\alpha y$ iff $x|_\alpha =_\alpha y|_\alpha$.

If $x|_\alpha = y|_\alpha$ then clearly $x|_\alpha =_\alpha y|_\alpha$.
Suppose now that $x|_\alpha =_\alpha y|_\alpha$ so that $x =_\alpha y$ also
holds. Since $x|_\alpha$
is the $\leq$-least element of $[x]_\alpha$  and $y|_\alpha$ is the $\leq$-least
element of $[y]_\alpha$, and since $[x]_\alpha = [y]_\alpha$,
we have that $x|_\alpha = y|_\alpha$.

Since $x =_\alpha x|_\alpha$ and $y =_\alpha y|_\alpha$,
we have (by the definition of $\sqsubset_\alpha$) that
$x \sqsubset_\alpha y$ iff $x|_\alpha \sqsubset_\alpha y|_\alpha$
and $x \sqsubseteq_\alpha y$ iff $x|_\alpha \sqsubseteq_\alpha y|_\alpha$.
\end{proof}
\begin{corollary}
\label{cor-sqsubset-first}
For all $x,y \in L$, $x \sqsubset y$ iff there is some $\alpha < \lord$ with
$x|_\alpha \sqsubset_\alpha y|_\alpha$. Moreover, $x \sqsubseteq y$
iff either there is some $\alpha < \lord$ with $x|_\alpha \sqsubset y|_\alpha$,
or $x =_\alpha y$ for all $\alpha < \lord$
(or equivalently, $x \sqsubseteq_\alpha y$, for all $\alpha < \lord$).
\end{corollary}

Notice that in the latter case, $x = y$, by Lemma~\ref{some-consequences}(a).

\begin{corollary}
\label{cor-sqsubset}
For all $x,y \in L$, $x \sqsubset y$ iff there is some $\alpha < \lord$
with $x|_\alpha \sqsubset_\alpha y$. Thus, $x \sqsubseteq y$ iff
either there is some $\alpha < \lord$ with $x|_\alpha \sqsubset_\alpha y$, or
$x|_\alpha =_\alpha y$ (or $x|_\alpha \sqsubseteq_\alpha y$)
for all $\alpha < \lord$.
\end{corollary}



\begin{lemma}
For all $x,y \in L$, $x = y$ iff $x|_\alpha = y|_\alpha$ for all $\alpha < \lord$.
\end{lemma}

\begin{proof}
The left to right direction is clear. Suppose that $x|_\alpha = y|_\alpha$ for all $\alpha < \lord$.
Then by Lemma~\ref{lem-alphaeq}, $x =_\alpha y$ for all $\alpha < \lord$, and thus,
by Axiom~\ref{axiom2}, $x = y$. \end{proof}

\begin{lemma}
\label{lem alpha < beta}
  Suppose that $x \in L$ and $\alpha < \beta < \lord$.
  Then $x|_\alpha \leq x|_\beta$
  and $x|_\alpha =_\alpha x|_\beta$.
\end{lemma}

\begin{proof}
 Since by Lemma~\ref{lem-slice} and Remark~\ref{remark0},
 $x|_\beta \in [x]_\beta \subseteq [x]_\alpha$,
 and since by Lemma~\ref{lem-slice} $x|_\alpha\in [x]_\alpha$ and
 $x|_\alpha$ is $\leq$-least in $[x]_\alpha$, we have both properties.
\end{proof}

\begin{lemma}\label{characterization-lemma}
 For every $x \in L$, $x = \bigvee_{\alpha < \lord} x|_\alpha$.
\end{lemma}

\begin{proof}
Let us denote $\bigvee_{\alpha < \lord} x|_\alpha$ by $y$.
Then by Lemma~\ref{lem alpha < beta},
also $y = \bigvee_{\beta \geq \alpha} x|_\beta$ for all fixed $\alpha < \lord$.
But by Lemma~\ref{lem alpha < beta}, $x|_\alpha =_\alpha x|_\beta$
for all $\beta$ with $\beta \geq \alpha$. Therefore, using Axiom~\ref{axiom4} we get
that for all $\alpha < \lord$, $x|_\alpha=_\alpha y$ and therefore $x =_\alpha y$. Since this holds for all $\alpha$,
by Axiom~\ref{axiom2} we get $x = y$. \end{proof}

\begin{lemma}\label{less-lemma}
For all $x,y \in L$, $x \leq y$ iff $x|_\alpha \leq y$ for all $\alpha < \lord$.
Moreover, if for all $\alpha < \lord$ there is some $\beta$ with $x|_\alpha
\leq y|_\beta$, then $x \leq y$. In particular, if
$x|_\alpha \leq y|_\alpha$ for all $\alpha < \lord$, then  $x \leq y$.
\end{lemma}

\begin{proof} Using the formula $x = \bigvee_{\alpha < \lord} x|_\alpha$ (Lemma~\ref{characterization-lemma}). \end{proof}

\subsection{Compatible Sequences}\label{compatible-sequences}
One can also view the elements of $L$ as least upper bounds of a
special kind of sequences, termed {\em compatible sequences}. The following
lemma will be used below:
\begin{lemma}
\label{lem-sup}
If $\alpha \leq \lord$ is an ordinal and $(x_\beta)_{\beta < \alpha}$ is a sequence
of elements of $L$ such that $x_\beta =_\beta x_\gamma$ and $x_\beta \leq x_\gamma$
whenever $\beta < \gamma <\alpha$, and if $x = \bigvee_{\beta < \alpha} x_\beta$,
then $x_\beta =_\beta x$ holds for all $\beta < \alpha$.
\end{lemma}

\begin{proof}
Indeed, this is clear when $\alpha$ is $0$ or a successor ordinal.
If $\alpha > 0$ is a limit ordinal, then
for each $\beta <\alpha$ it holds that $x = \bigvee_{\beta \leq \gamma < \alpha}x_\gamma$
and $x_\beta =_\beta x$ for all $\beta \leq \gamma < \alpha$.
Thus, by Axiom~\ref{axiom4}, $x =_\beta x_\beta$ for all $\beta < \alpha$.
\end{proof}

\begin{definition}\label{compatible-sequence}
A sequence $(x_\alpha)_{\alpha<\lord}$ of elements of $L$ is called \emph{compatible}
if each $x_\alpha$ is the $\leq$-least element of $[x_\alpha]_\alpha$
and if $x_\alpha =_\alpha x_\beta$ for all $\alpha <\beta$.
\end{definition}

Note that the above definition implies that $x_\alpha \leq x_\beta$ for all
$\alpha < \beta$, since under these conditions, if $\alpha < \beta$,
then $[x_\beta]_\beta \subseteq [x_\alpha]_\alpha$, and thus
$x_\alpha \leq x_\beta$, since $x_\alpha$ is $\leq$-least in $[x_\alpha]_\alpha$.
Also, by Lemma~\ref{lem-slice}, $x_\alpha$ is a $\sqsubseteq_{\alpha + 1}$-least
element of $[x_\alpha]_\alpha$.

\begin{example}
Consider again the standard model presented at the end of Section~\ref{models-section}.
We give an example of a compatible sequence of interpretations, which may
facilitate the understanding of the remaining results of this subsection.
Let $Z=\{p_i : i< \omega\}$ be a set of propositional variables. We define
the following compatible sequence of interpretations $(I_\alpha)_{\alpha<\Omega}$:
\[
\begin{array}{ll}
(I_\alpha)(p_i) = \left\{\begin{array}{ll}
                        T_i,          &\mbox{if $i\leq \alpha$}\\
                        F_{\alpha+1}, &\mbox{otherwise}
                  \end{array}
                  \right.
\end{array}
\]
It can be easily checked that $(I_\alpha)_{\alpha<\Omega}$ satisfies the
requirements of Definition~\ref{compatible-sequence} and is therefore a
compatible sequence of interpretations. Let $I=\bigvee_{\alpha<\Omega}  I_\alpha$.
Then, it is easy to verify that $I(p_i) = T_i$ for all $i< \omega$. Moreover, notice that
$I_\alpha = I|_\alpha$, for all $\alpha<\Omega$. This is not a coincidence, as the following
Lemma illustrates.
\end{example}

\begin{lemma}
Suppose that $(x_\alpha)_{\alpha<\lord}$ is a compatible sequence and let
$x = \bigvee_{\alpha<\lord}  x_\alpha$. Then $x|_\alpha = x_\alpha$
for all $\alpha<\lord$.
\end{lemma}

\begin{proof}
We know that $x_\alpha =_\alpha x_\beta$ and $x_\alpha \leq x_\beta$
for all $\alpha < \beta$. Thus, by Lemma~\ref{lem-sup}, $x =_\alpha x_\alpha$
for all $\alpha < \lord$. But also $x =_\alpha x|_\alpha$ (since $x|_\alpha \in [x]_\alpha$)
and therefore $x|_\alpha =_\alpha x_\alpha$
for all $\alpha < \lord$. Thus, $[x_\alpha]_\alpha = [x|_\alpha]_\alpha$,
and since both $x|_\alpha$ and $x_\alpha$ are $\leq$-least elements of this set,
$x|_\alpha = x_\alpha$.
\end{proof}

The following fact is an immediate corollary of the previous lemmas.

\begin{corollary}
There is a bijection between $L$ and the set of all compatible sequences
$(x_\alpha)_{\alpha < \lord}$. This bijection maps $x\in L$ to the compatible sequence
$(x|_\alpha)_{\alpha < \lord}$, and a compatible sequence
$(x_\alpha)_{\alpha < \lord}$ to $\bigvee_{\alpha < \lord} x_\alpha$.
\end{corollary}

We can now restate Corollaries~\ref{cor-sqsubset-first} and~\ref{cor-sqsubset}
as-well-as Lemma~\ref{less-lemma} under the view of compatible sequences:
\begin{corollary}
Let $x,y\in L$ and let $(x_\alpha)_{\alpha < \lord}$ and $(y_\alpha)_{\alpha < \lord}$ be the
corresponding compatible sequences. Then, $x \sqsubset y$ iff there is some $\alpha < \lord$ with
$x_\alpha \sqsubset_\alpha y_\alpha$. Moreover, $x \sqsubseteq y$
iff either there is some $\alpha < \lord$ with $x_\alpha \sqsubset y_\alpha$,
or $x =_\alpha y$ for all $\alpha < \lord$
(or equivalently, $x \sqsubseteq_\alpha y$, for all $\alpha < \lord$).
\end{corollary}
\begin{corollary}
\label{cor-sqsubset-compatible}
Let $x,y\in L$ and let $(x_\alpha)_{\alpha < \lord}$ and $(y_\alpha)_{\alpha < \lord}$ be the
corresponding compatible sequences. Then, $x \sqsubset y$ iff there is some $\alpha < \lord$
with $x_\alpha \sqsubset_\alpha y$. Thus, $x \sqsubseteq y$ iff
either there is some $\alpha < \lord$ with $x_\alpha \sqsubset_\alpha y$, or
$x_\alpha =_\alpha y$ (or $x_\alpha \sqsubseteq_\alpha y$)
for all $\alpha < \lord$.
\end{corollary}
\begin{lemma}
Let $x,y\in L$ and let $(x_\alpha)_{\alpha < \lord}$ and $(y_\alpha)_{\alpha < \lord}$ be the
corresponding compatible sequences. Then, $x \leq y$ iff $x_\alpha \leq y$ for all $\alpha < \lord$.
Moreover, if for all $\alpha < \lord$ there is some $\beta$ with $x_\alpha
\leq y_\beta$, then $x \leq y$. In particular, if
$x_\alpha \leq y_\alpha$ for all $\alpha < \lord$, then  $x \leq y$.
\end{lemma}

We now slightly generalize the notion of compatible sequence:
\begin{definition}
Let $(x_\beta)_{\beta < \alpha}$, $\alpha \leq \lord$, be a sequence of elements of
$L$. We call this sequence a \emph{partial compatible sequence}
if each $x_\beta$ is the $\leq$-least
element of $[x_\beta]_\beta$, and if $x_\beta =_\beta x_\gamma$
for all $\beta < \gamma$.
\end{definition}

Note that the above definition implies
that $x_\beta \leq x_\gamma$ for all $\beta < \gamma <\alpha$.
Moreover, for each $\beta < \alpha$, $x_\beta$ is a
$\sqsubseteq_{\beta + 1}$-least element of $[x_\beta]_\beta$.
Notice also that when $\alpha$ is equal to
$\lord$ then we actually have a compatible sequence.

\begin{lemma}
\label{lem-partial}
Suppose that $(x_\beta)_{\beta < \alpha}$ is a partial compatible
sequence where $\alpha < \lord$. Define $x = \bigvee_{\beta < \alpha} x_\beta$.
Let $y_\delta = x$ for all $\alpha \leq \delta <\lord$ and $y_\beta = x_\beta$
if $\beta < \alpha$. Then $x$ is the $\leq$-least
element of $(x]_\alpha$, which is also necessarily a
$\sqsubseteq_\alpha$-least element of $(x]_\alpha$.
Moreover, the sequence $(y_\delta)_{\delta < \lord}$
is compatible.
\end{lemma}

\begin{proof}
By Lemma~\ref{lem-sup}, $x =_\beta x_\beta$ for all $\beta < \alpha$.
Consider the set $(x]_\alpha$.
We want to show that
$x$ is the $\leq$-least element of $(x]_\alpha$.

First, $x \in (x]_\alpha$.
Second, if $z \in (x]_\alpha$ then $z =_\beta x =_\beta x_\beta$ for all
$\beta < \alpha$, hence $x_\beta \leq z$ for all $\beta < \alpha$ and
$x =\bigvee_{\beta < \alpha} x_\beta \leq z$.


To complete the proof of the fact that $(y_\delta)_{\delta < \lord}$ is compatible,
it suffices to show that $y_\delta$ is the $\leq$-least element of $[x]_\delta$
for all $\delta \geq \alpha$. So let $\delta \geq \alpha$.
Then $x \in [x]_\delta \subseteq (x]_\alpha$. Since $x$ is
$\leq$-least in $(x]_\alpha$ and $x \in [x]_\delta$, it is
also $\leq$-least in $[x]_\delta$.
\end{proof}

\begin{corollary}\label{aux-corol}
For each partial compatible sequence $(x_\beta)_{\beta < \alpha}$, $\alpha \leq \lord$,
$x = \bigvee_{\beta < \alpha} x_\beta$ is the unique element of
$L$ with $x|_\beta = x_\beta$  for $\beta < \alpha$ and $x|_\delta = x$
for all $\delta$ with $\alpha \leq \delta < \lord$.
\end{corollary}

\begin{corollary}
\label{cor-partsup}
Suppose that $(x_\beta)_{\beta < \alpha}$, $\alpha \leq \lord$, is a partial compatible
sequence and $x = \bigvee_{\beta < \alpha} x_\beta$.
If $z \in L$ with $x_\beta \sqsubseteq_\beta z$ for all $\beta < \alpha$, then
$x \sqsubseteq z$.
\end{corollary}

\begin{proof}
If $\alpha = \lord$ then $(x_\beta)_{\beta < \alpha}$ is a compatible sequence. Since
$x_\beta \sqsubseteq_\beta z$ for all $\beta < \alpha$, we have by Lemma~\ref{some-consequences}(a)
that $x=y$.

Assume now that $\alpha < \lord$. By Corollary~\ref{aux-corol} we
know that $x|_\delta = x_\delta$ if $\delta < \alpha$,
otherwise $x|_\delta = x$. Thus, by Corollary~\ref{cor-sqsubset},
our claim is clear if $x|_\beta \sqsubset_\beta z$ (ie., $x_\beta \sqsubset_\beta z$)
for some $\beta < \alpha$.

Suppose that $x_\beta =_\beta z$ for all $\beta < \alpha$.
Then $z \in (x]_\alpha$. By Lemma~\ref{lem-partial}
we know that $x$ is the $\leq$-least
element and a $\sqsubseteq_\alpha$-least element of $(x]_\alpha$.
Thus, $x \sqsubseteq_\alpha z$. If $x \sqsubset_\alpha z$
then we are done.

Suppose that $x =_\alpha z$.
If $x =_\gamma z$ for all $\gamma \geq \alpha$, then $x = z$.
Otherwise let $\gamma > \alpha$ denote the
least ordinal such that
$x \neq_\gamma z$. We have that $x|_\gamma = x$.
Also, $x \in (x]_\gamma \subseteq (x]_\alpha$.
Since $x$ is $\leq$-least in $(x]_\alpha$, $x$ is
also $\leq$-least in $(x]_\gamma$. Thus, $x \sqsubseteq_\gamma z$,
since the $\leq$-least element is also a $\sqsubseteq_\gamma$-least element
of $(x]_\gamma$.
 But $x \neq_\gamma z$, thus $x \sqsubset_\gamma z$,
proving $x \sqsubseteq z$.
\end{proof}

\section{$L$ is a Complete Lattice}\label{lattice-section}
 In this section we demonstrate that the partial order $(L,\sqsubseteq)$ is
actually a complete lattice. We use the following convenient definition
of complete lattices:
\begin{definition}~\cite{DP02}[Definition 2.4, page 34 and Theorem 2.31, page 47]
A partially ordered set $(L,\sqsubseteq)$ is called a {\em complete lattice}
if $L$ has a bottom element and every non-empty subset $X\subseteq L$
has a least upper bound in $L$.
\end{definition}

The above definition is equivalent to the more usual one (which specifies
that a partially ordered set $(L,\sqsubseteq)$ is a complete lattice if
every subset $X$ of $L$ has a least upper bound and a greatest lower bound
in $L$).

Therefore, in our case it suffices to show that $\bot \sqsubseteq x$
for all $x \in L$ and that every nonempty subset $X$ of $L$ has a
least upper bound $x_\infty \in L$ with respect to $\sqsubseteq$.
The construction of $x_\infty$ proceeds in stages. At each stage $\alpha$,
an approximation $x_\alpha$ of $x_\infty$ is constructed. Intuitively, $x_\alpha$
is an upper bound of the elements of $X$ if we restrict our comparison to
the stages that are less than or equal to $\alpha$. The limit of this construction
is the desired element $x_\infty$.

The main result of this section implies (as a special case) that the set of
interpretations of a logic program equipped with the relation $\sqsubseteq$,
forms a complete lattice. The bottom element of this lattice
is the interpretation that assigns to each propositional variable the value $F_0$
and the top element is the one that assigns to every variable the value $T_0$. At
a first reading of the proof of the following theorem, it may be easier to think
of the standard model (instead of arbitrary models of Axioms 1-4).

\begin{theorem}\label{complete-lattice-theorem}
$(L,\sqsubseteq)$ is a complete lattice.
\end{theorem}

\begin{proof}
First we show that $\perp$ is the $\sqsubseteq$-least element of $L$, ie. that
$\perp \sqsubseteq x$ holds for all $x \in L$.
If $x \neq \perp$,
there is a least ordinal $\alpha < \lord$ such that $x \not =_\alpha \perp$
(since we know that $\perp\kern-.4em|_\alpha = \perp$ for all $\alpha < \lord$).
Thus $x =_\beta \perp$ for all $\beta < \alpha$, ie., $x \in (\perp]_\alpha$.
We have noted (Remark~\ref{rem-remark1}) that $\perp \sqsubseteq_\alpha z$ holds for all $z \in (\perp]_\alpha$.
In particular, $\perp \sqsubseteq_\alpha x$. But $x \not =_\alpha \perp$, thus
$\perp \sqsubset_\alpha x$. It follows that $\perp \sqsubset x$.

We show that every non-empty subset $X$ of $L$ has a least upper bound
$x_\infty$ with respect to $\sqsubseteq$. Let $X \subseteq L = (\perp]_0$,
$X \neq \emptyset$. For each $\alpha < \lord$, we define $x_\alpha \in L$ and
$X_\alpha,Y_\alpha  \subseteq X$. Let $Y_0= X$ and  $x_0 = \bigsqcup_0 Y_0$ and
$$X_0 = \{x \in Y_0 : x =_0 x_0\} = \{x \in X : x =_0 x_0\}$$
For each nonzero ordinal $\alpha < \lord$,
we define $x_\alpha$, $X_\alpha$ and $Y_\alpha$ as follows:
\begin{eqnarray*}
Y_\alpha & = & \bigcap_{\beta < \alpha} X_\beta\\
x_\alpha & = & \left\{ \begin{array}{ll}
                        \bigvee_{\beta < \alpha}x_\beta & \mbox{if $Y_\alpha = \emptyset$}\\
                        \bigsqcup_\alpha Y_\alpha       & \mbox{if $Y_\alpha \neq \emptyset$}
                       \end{array}
               \right.\\
X_\alpha & = & \{x \in Y_\alpha : x =_\alpha x_\alpha\}
\end{eqnarray*}
Notice that if $Y_\alpha =\emptyset$ then $X_\alpha=\emptyset$. From the above definitions it is easy
to see that:
$$Y_\alpha = \bigcap_{\beta < \alpha} X_\beta =
\{x \in X : \forall \beta < \alpha\ x =_\beta x_\beta\}$$
%
and:
$$X_\alpha = \{x \in Y_\alpha : x =_\alpha x_\alpha\}= \{x \in X : \forall \beta \leq \alpha\ x =_\beta x_\beta\}$$
%
It is clear from the above that $X_\beta \supseteq X_\alpha$ whenever $\beta < \alpha$.

We show by induction on $\alpha < \lord$ that $(x_\beta)_{\beta < \alpha}$ is a
partial compatible sequence. To this end, it is sufficient to prove that
$x_\beta =_\beta x_\alpha$ for all $\beta < \alpha$,
since it follows then that each $x_\beta$ is $\leq$-least in $[x_\beta]_\beta$ (Lemma~\ref{lem-lem0}
and Lemma~\ref{lem-partial}).

We argue by induction on $\alpha$.
When $\alpha = 0$ our claim is obvious. Suppose now that $\alpha > 0$
and the claim is true for all ordinals less than $\alpha$. Let $\beta < \alpha$.
We distinguish two cases:
\begin{itemize}
\item $Y_\alpha \neq \emptyset$. Let $y$ be a fixed element of $Y_\alpha$.
      Since $y \in X_\beta$, we have that $x_\beta =_\beta y$.
      Also, $y \sqsubseteq_\alpha x_\alpha$,
      since $x_\alpha = \bigsqcup_\alpha Y_\alpha$.
      Thus, by Lemma~\ref{some-consequences}(d), $x_\beta =_\beta x_\alpha$.

\item $Y_\alpha = \emptyset$. Then by definition, $x_\beta \leq x_\alpha$.
      To complete the proof, it suffices to show that $x_\beta =_\beta x_\alpha$
      in the case when $\alpha$ is the least ordinal such that $Y_\alpha = \emptyset$.
      But $x_\beta =_\beta x_\gamma$ for all $\beta \leq \gamma < \alpha$
      and thus, by Lemma~\ref{lem-sup}, $x_\beta =_\beta \bigvee_{\beta \leq \gamma < \alpha}x_\gamma = x_\alpha$.
\end{itemize}

We have proved that for all $\alpha<\lord$, $(x_\beta)_{\beta < \alpha}$ is a partial compatible sequence.
Therefore, $(x_\alpha)_{\alpha<\lord}$ is a compatible sequence.
In particular, $x_\alpha \leq x_\beta$ whenever $\alpha < \beta$.

Since $(x_\alpha)_{\alpha < \lord}$ is a compatible sequence, it
determines a unique element of $L$, $x_\infty = \bigvee_{\alpha < \lord} x_\alpha$.
Note that for all $\alpha < \lord$, $x_\infty =_\alpha x_\alpha$. Let $X_\infty = \bigcap_{\alpha < \lord} X_\alpha$.
Our aim is to prove that $x_\infty$ is the least upper bound of $X$ with respect to the relation $\sqsubseteq$.
Moreover, we prove that either $X_\infty = \emptyset$ or $X_\infty = \{x_\infty\}$,
ie., $X_\infty \subseteq \{x_\infty\}$.

{\em Proof of $X_\infty \subseteq \{x_\infty\}$.}
Suppose that $y \in X_\infty$. Then for all $\alpha$,
$x_\infty =_\alpha x_\alpha =_\alpha y$. It follows that $x_\infty = y$,
since by Axiom 2 it holds that $\bigcap_{\alpha < \lord} =_\alpha$ is the equality relation.

{\em Proof of $X \sqsubseteq x_\infty$.}
Let $y \in X$.
There are two cases, either $y \in X_\infty$ or $y \not\in X_\infty$. If $y \in X_\infty$
then $y = x_\infty$ and clearly $y \sqsubseteq x_\infty$. If $y \not\in X_\infty$,
then there is a least
ordinal $\alpha$ less than $\lord$ such that $y \not\in X_\alpha$.
We have $y \in Y_\alpha$ and thus $y \sqsubseteq_\alpha x_\alpha$,
since $x_\alpha = \bigsqcup_\alpha Y_\alpha$.
But $y \not=_\alpha x_\alpha$ since $y \not\in X_\alpha$.
Thus $y \sqsubset_\alpha x_\alpha$
and  $y \sqsubseteq x_\infty$.

To complete the proof, it remains to show that $x_\infty$ is the {\em least}
among the upper bounds of $X$. Therefore suppose that for some $z$, $X \sqsubseteq z$.
We need to prove that $x_\infty\sqsubseteq z$. This is clear when
$X_\infty \neq \emptyset$,
since in this  case $x_\infty \in X$. Suppose now that
$X_\infty = \emptyset$, and let
$\alpha$ denote the least ordinal that is less than $\lord$ such that $X_\alpha = \emptyset$,
if such an ordinal exists, otherwise let $\alpha = \lord$.
We prove by induction on
$\beta < \alpha$ that either $x_\beta \sqsubseteq_\beta z$,
or there is some $\gamma < \beta$ with $x_\gamma \sqsubset_\gamma z$.

When $\beta = 0$, $x_0 = \bigsqcup_0 X$. Since $X \sqsubseteq z$,
by Lemma~\ref{inclusion-lemma} we have $X \sqsubseteq_0 z$. Thus, $x_0\sqsubseteq_0 z$.
Suppose that $\beta > 0$ and the claim is true for all
ordinals less than $\beta$. If there is some $\gamma < \beta$
with $x_\gamma \sqsubset_\gamma z$ then we are done. Otherwise
$x_\gamma =_\gamma z$ for all $\gamma < \beta$.
Now $x_\beta = \bigsqcup_\beta Y_\beta$, where $Y_\beta = \bigcap_{\gamma < \beta} X_\gamma
= \{y \in X: \forall \gamma < \beta\ x_\gamma =_\gamma y\}$.
Since $x_\gamma =_\gamma z$ for all $\gamma < \beta$ and $Y_\beta \sqsubseteq z$,
it follows by Lemma~\ref{lem-lem1} that $x_\beta \sqsubseteq_\beta z$.

To complete the proof of the fact that $x_\infty \sqsubseteq z$, first note
that by Corollary~\ref{cor-sqsubset-compatible},
if there is some $\beta < \alpha$ with $x_\beta \sqsubset_\beta z$,
then $x_\infty \sqsubset z$.
If $x_\beta =_\beta z$ for all $\beta < \alpha$, then we use
Corollary~\ref{cor-partsup} to conclude that $x_\infty \sqsubseteq z$.
\end{proof}

\section{Certain Models of the Axioms}\label{models-section}
In this section we investigate models of the axioms introduced in Section~\ref{axioms-section}.
Apart from the set of infinite-valued interpretations of logic programs, it turns out
that there exist certain alternative models with different structures and properties.
\subsection{The Standard Model}\label{standard-model}
As it was demonstrated at the end of Section~\ref{axioms-subsection}, the set of
infinite-valued interpretations together with the relations introduced in
Section~\ref{infinite-valued-section}, form a model of the Axioms 1-4.
As we have already mentioned, this particular model will be referred as
``{\em the standard model}'' throughout the paper. In the following, we
give some more facts regarding this model.

Apart from Axioms 1-4, the standard model also satisfies the following Axiom 5. As we are going to see in
Subsection~\ref{non-standard-product-model}, not all models of Axioms 1-4 satisfy this axiom.
\begin{framed}
\begin{axiom}\label{axiom5}
Let $x,y \in L$ and $\alpha < \lord$. If $x \leq y$ and $x =_\beta y$ for all $\beta < \alpha$,
then $x \sqsubseteq_\alpha y$.
\end{axiom}
\end{framed}

It is straightforward to verify that Axiom 5 holds in the standard model.
Given Axiom 5, we can prove:
\begin{proposition}
\label{prop-ax5cons}
Let $L$ be a model and $x,y\in L$ with $x\leq y$. If Axiom 5 holds in $L$,
then $x \sqsubseteq y$.
\end{proposition}

\begin{proof}
If $x=y$ the proposition obviously holds.
Suppose that $x < y$. Let $\alpha$ denote the least
ordinal such that $x \neq_\alpha  y$. Then $x =_\beta y$
for all $\beta < \alpha$. Since $x < y$, we have by Axiom 5 that
$x \sqsubseteq_\alpha y$. But $x\neq_\alpha  y$, thus
$x \sqsubset_\alpha y$, so that $x \sqsubset y$.
\end{proof}

\subsection{The Model of Truth Values}\label{model-of-truth-values}
When $Z$ has a single element, the infinite set $V$ of truth values introduced
in Section~\ref{infinite-valued-section} may be identified with $V^Z$,
and this implies that $V$ is also a model of the axioms. However, since the model
of truth values will be used several times in the following sections,
it is more convenient to define the relations $\sqsubseteq_\alpha$ directly
on $V$ and not use the isomorphism with $V^Z$ when $Z$ is a singleton
(from a mathematical point of view, the two approaches are equivalent).
For reasons of completeness, we will give the proofs that $V$ is
a model (despite the fact that these proofs can be retrieved from the
corresponding ones for $V^Z$ when $Z$ is a singleton).

The relations $\sqsubseteq_\alpha$ on $V$ for $\alpha < \Omega$
are defined by $x \sqsubseteq_\alpha y$ iff either $\order(x) =
\order(y)  < \alpha$ and $x = y$, or $\order(x),\order(y) \geq \alpha$,
and either $\order(x),\order(y) > \alpha$ or $x = F_\alpha$ or $y = T_\alpha$.
Of course, we have $x \sqsubseteq_\alpha y$ iff either
$x = y$, or $\order(x),\order(y) \geq \alpha$, and either
$\order(x),\order(y) > \alpha$ or $x = F_\alpha$ or $y = T_\alpha$.
Thus, $x =_\alpha y$ iff $x = y$ or $\order(x),\order(y) > \alpha$.
It is clear that Axioms 1, 2 and 4 hold. To show that Axiom 3 also holds, let $x\in V$,
$\alpha < \Omega$ and let $X \subseteq (x]_\alpha$ be a nonempty set.
We distinguish two cases. If $\order(x) < \alpha$, then $(x]_\alpha = \{x\}$
and $\bigsqcup_\alpha X = x$ satisfies Axiom 3. If on the other hand
$\order(x) \geq \alpha$, then
$(x]_\alpha = \{y : \order(y) \geq \alpha\}$.
We consider two subcases: if $\bigvee X = T_\alpha$ or $\bigvee X = F_\alpha$ then
we take $\bigsqcup_\alpha X = \bigvee X$, which is easily seen to satisfy Axiom 3
(notice that $\bigvee X = T_\alpha$ happens exactly when $T_\alpha \in X$,
and $\bigvee X = F_\alpha$ happens when $X = \{F_\alpha\}$).
In any other case, we take $\bigsqcup_\alpha X = F_{\alpha+1}$,
which again satisfies Axiom 3.

To show that Axiom 5 also holds, let $x,y\in V$
with $x\leq y$ and $x=_\beta y$ for all $\beta < \alpha$, where
$\alpha < \Omega$. Then either $x = y$ or $\order(x),\order(y) \geq \alpha$.
In the former case, $x \sqsubseteq_\alpha y$ clearly holds.
So suppose that $x \neq y$ and $\order(x),\order(y) \geq \alpha$.
If $x = F_\alpha$ or $y = T_\alpha$, then $x \sqsubseteq_\alpha y$.
Otherwise $\order(x),\order(y) > \alpha$ and $x \sqsubseteq_\alpha y$ again.

\subsection{The Product Model}\label{product-model}
Let $I$ be an index set. Suppose that for all $i \in I$, $L_i$ satisfies Axioms 1-4. We define a
new model $L$ on the cartesian product $\prod_{i \in I}L_i$.
For simplicity we overload our notation by using the same symbols $\leq$, $\sqsubseteq_\alpha$
and $\sqsubset_\alpha$ for the ordering relations in all $L_i$ and in $L$.
For any $f,g \in L$ we define $f\leq g$ iff $f(i) \leq g(i)$
for all $i \in I$. Moreover, for each $\alpha < \kappa$,
we define $f \sqsubseteq_\alpha g$ iff $f(i) \sqsubseteq_\alpha g(i)$
for each $i \in I$. Thus, $f \sqsubset_\alpha g$
iff $f \sqsubseteq_\alpha g$ and there is some $i\in I$
with $f(i) \sqsubset_\alpha g(i)$.

We claim that $L$ is a model. Indeed, it is clear that Axioms 1 and 2 hold.
In order to prove Axiom 3, suppose that $h\in L$ and
$X\subseteq (h]_\alpha$, where $\alpha < \lord$.
Then $X(i) = \{f(i) : f \in X \} \subseteq (h(i)]_\alpha$
in $L_i$, and thus $g(i) = \bigsqcup_\alpha X(i)$ is in $L_i$.
Suppose that $h'\in (h]_\alpha$ with $X \sqsubseteq_\alpha h'$.
Then $h'(i) \in (h(i)]_\alpha$ and $X(i) \sqsubseteq_\alpha h'(i)$
for all $i \in I$. Since Axiom 3 holds in each $L_i$,
we obtain that $g(i) \sqsubseteq_\alpha h'(i)$ and
$g(i) \leq h'(i)$. Since this holds for all $i \in I$,
we conclude that $g \sqsubseteq_\alpha h'$ and $g \leq h'$.
In order to verify Axiom 4, suppose now that $X \subseteq L$
is not empty and $f =_\alpha g$ holds for all $f\in X$,
where $g \in L$ and $\alpha < \kappa$. Then $f(i) =_\alpha g(i)$
for all $f\in X$ and $i \in I$, so that $\bigvee X(i) =_\alpha g(i)$
for all $i\in I$. It follows that $\bigvee X =_\alpha g$.
We also note that if each $L_i$ satisfies Axiom 5,
then so does $L$.

As an application of the above model construction
operation, observe that the standard model
could have been obtained by the fact that the set $V$ of truth values is a model
(see Subsection~\ref{model-of-truth-values}); therefore
the set $V^Z$, where $Z$ is any set of propositional variables, is also
a model. Moreover, since $V$ satisfies Axiom 5, $V^Z$ also does.
More generally we have the following lemma, whose proof is straightforward:

\begin{lemma}
\label{lem-ax5}
Suppose that $L_i$ is a model satisfying Axiom 5 for each $i \in I$.
Then the product model $L = \prod_{i\in I} L_i$ also satisfies Axiom 5.
\end{lemma}

\subsection{A Non-Standard Product Model}\label{non-standard-product-model}
%
Suppose that for each $\alpha < \lord$, $(L_\alpha,\leq)$
is a complete lattice. Let  $L$ be the cartesian
product $\prod_{\alpha < \lord} L_\alpha$. Then $(L,\leq)$,
equipped with the pointwise ordering
$$x \leq y \quad {\rm iff}\quad \forall\alpha < \lord\ x(\alpha)\leq y(\alpha)$$
is a complete lattice. Suppose now that for each $\alpha < \lord$, $(L_\alpha,\preceq)$
is another complete lattice with underlying set $L_\alpha$. We
use the orderings $\preceq$ to define preorderings
$\sqsubseteq_\alpha$ on $L$. For each $x,y\in L$ and $\alpha < \lord$,
we define $x \sqsubseteq_\alpha y$ iff $x(\beta) = y(\beta)$
for all $\beta< \alpha$ and $x(\alpha) \preceq y(\alpha)$.
In the following we will refer to the construction just described
as the {\em nonstandard product construction}.

The following lemma is straightforward to establish:
\begin{lemma}
When $\preceq$ coincides with $\leq$ for each $L_\alpha$ we get a
model that satisfies Axiom 5.
\end{lemma}

We are therefore now interested in the case where $\preceq$ and $\leq$
do not coincide. It is clear that Axioms 1 and 2 hold. Axiom 4 also holds.
Indeed, let $X \subseteq L$ be a nonempty set, $y \in L$,
and let $\alpha < \lord$.  Suppose that
$y =_\alpha x$ for all $x \in X$. Then
$x(\alpha) = y(\alpha)$ for all $x \in X$
and thus $(\bigvee X)(\alpha) = \bigvee_{x \in X} x(\alpha) = y(\alpha)$,
ie., $y =_\alpha \bigvee X$. Regarding Axiom 3, we have the
following lemma:

\begin{lemma}
Axiom 3 holds in $L$ iff for all $\alpha < \lord$ and $a,b \in L_\alpha$,
if $a \preceq b$ then $a \leq b$, ie., when $\leq$ extends $\preceq$
for all $L_\alpha$.
\end{lemma}

\begin{proof}
Assume first that Axiom 3 holds, and let $a,b\in L_\alpha$ for some
fixed $\alpha < \lord$ with $a \preceq b$. Let $x,y \in L$
be such that $x(\beta)=y(\beta)$ is the least element of $L_\beta$ with respect
to the ordering $\leq$ for all $\beta \neq  \alpha$, $\beta < \lord$,
and $x(\alpha) = a$, $y(\alpha) = b$. Then $x \sqsubseteq_\alpha y$.
Since Axiom 3 holds by assumption, there is some $z$ with $\{x,y\} \sqsubseteq_\alpha z$
such that whenever $\{x,y\}\sqsubseteq_\alpha z'$ then $z \sqsubseteq_\alpha z'$ and $z \leq z'$.
In particular, let $z' = y$. Then we must have $z \sqsubseteq_\alpha y$ and $z \leq y$,
and of course also $y \sqsubseteq_\alpha z$, so that $y =_\alpha z$. Since
$z \leq y$, we have that $z(\beta)$ is the least element of $L_\beta$
for all $\beta \neq \alpha$. Since $y =_\alpha z$, $z(\alpha) = b$.
Thus $z = y$. Suppose now that $b \preceq c$ holds in $L_\alpha$. Then let $z'\in L$
such that $z'(\alpha)$ is the $\leq$-least element of $L_\beta$ for all $\beta \neq \alpha$,
and let $z'(\alpha) = c$. Then $\{x,y\} \sqsubseteq_\alpha z'$ and thus $z \sqsubseteq_\alpha z'$
and $z \leq z'$, so that $b \leq c$. We have thus established that if $a \preceq  b \preceq c$
in $L_\alpha$, then $b \leq c$. In particular, let $a$ be the $\preceq$-least element
of $L_\alpha$. Then we obtain that $b \leq c$ whenever $b \preceq c$.


In order to prove the reverse direction,
suppose now that for each $\alpha < \lord$,
the ordering $\leq$ of $L_\alpha$ extends the ordering $\preceq$.
Let $x_0 \in L$ and $\alpha < \lord$. Moreover, let $X\subseteq (x_0]_\alpha$.
Then define $y(\beta) = x_0(\beta)$ for all $\beta < \alpha$
and let $y(\beta)$ be the $\leq$-least element of $L_\beta$ if $\alpha < \beta < \lord$;
finally, let $y(\alpha)$ be the supremum of the set $\{x(\alpha) : x \in X\}$
in the lattice $(L_\alpha,\preceq)$.
Then $y$ is the element $\bigsqcup_\alpha X$ specified by Axiom 3.
Indeed, it is clear that $X \sqsubseteq_\alpha y$.
Suppose that $X \sqsubseteq_\alpha z$ where $z \in (x_0]_\alpha$.
Then $x(\alpha)  \preceq z(\alpha)$ for
all $x \in X$, and thus $y(\alpha)  \preceq z(\alpha)$.
Since $x_0(\beta) = z(\beta)$ for $\beta < \alpha$, we have that
$y \sqsubseteq_\alpha z$.
Since $\leq$ extends $\preceq$ on $L_\alpha$, $y(\alpha)  \leq  z(\alpha)$.
Since $y(\beta)$ is the $\leq$-least element of $L_\beta$ for
all $\alpha < \beta < \lord$, and since $y(\beta) = x_0(\beta) = z(\beta)$
for all $\beta < \alpha$, we also have $y \leq z$.
\end{proof}

Notice that when $\leq$ is an extension of $\preceq$ on $L_\alpha$, then
the $\leq$ and $\preceq$-least elements of $L_\alpha$
are the same. Similarly, the $\leq$ and $\preceq$-greatest elements
are also the same.

As it was demonstrated above, $L$ is a
model of Axioms 1-4 provided that $\leq$ extends $\preceq$ for all $L_\alpha$.
However, this model does not always satisfy Axiom 5:
\begin{lemma}
There exists an instance of the non-standard product construction that does not
satisfy Axiom 5.
\end{lemma}
\begin{proof}
For each $\alpha< \lord$, let $(L_\alpha,\leq)$ be the $4$-element chain,
and $(L_\alpha,\preceq)$ the $4$-element lattice that is not a chain,
with the same least and greatest elements. Then we have a model that violates
Axiom 5. In fact, it violates the implication $x \leq y \Rightarrow x \sqsubseteq y$
of Proposition~\ref{prop-ax5cons}.
\end{proof}

\subsection{Some Further Consequences of Axiom 5}
\label{subsection-further}
In this subsection we discuss some further consequences of Axiom 5.
The material of this subsection will not be further used in the rest
of the paper, and is included for completeness reasons.

When $L$ is a model and $\alpha < \lord$, let us denote the set
$\{x|_\alpha : x \in L\}$ by $L_\alpha$. By Lemma~\ref{lem-slice added}
it holds $L_\alpha = \{x \in L : x = \bigsqcup_\alpha \{ x \} \} =
\{ \bigsqcup_\alpha \{x\}: x \in L\}$.

For example, when $L = V$ and $\alpha < \Omega$, then $L_\alpha =
\{F_0,\ldots,F_{\alpha+1},T_\alpha,\ldots,T_0\}$. And when $Z$ is any set
and $L = V^Z$, then $L_\alpha$ is the collection of all functions
$f: Z \to V$ with $f(Z) \subseteq \{F_0,\ldots,F_{\alpha+1},T_\alpha,\ldots,T_0\}$.
More generally, when $L$ is the product model $\prod_{i \in I}L_i$,
then $L_\alpha = \prod_{i \in I}(L_i)_\alpha$, due to the fact that
$(\bigsqcup_\alpha \{x\})(i) = \bigsqcup_\alpha\{x(i)\}$
for all $x \in L$ and $i \in I$.

\begin{proposition}\label{prop-ax5-2}
Suppose that $L$ is a model, $\alpha < \kappa$ and $x,y \in L$.
If $x \in L_\alpha$ and $x \sqsubseteq_\alpha y$ then $x \leq y$.
If Axiom 5 holds and $x \in L_\alpha$ and  $x =_\beta y$ for all $\beta < \alpha$,
then $x \leq y$ iff $x \sqsubseteq_\alpha y$.
\end{proposition}
\begin{proof}
Suppose that $x\in L_\alpha$
with $x \sqsubseteq_\alpha y$. Then since $x = \bigsqcup_\alpha \{x\}$, it holds
$x \leq y$. The second claim is immediate from Axiom 5 and the first claim.
\end{proof}
\begin{proposition}
\label{prop-ax5-1}
Suppose that $L$ is a model satisfying Axiom 5 and let $\alpha < \kappa$. Assume
that for each $i \in I$, $x_i \in L_\alpha$, where $I$ is a nonempty set. Let
$X = \{x_i: i \in I\}$. Assume that $x_i =_\beta x_j$
for all $i,j\in I$ and all $\beta < \alpha$. Then $\bigvee X \in L_\alpha$
and $\bigvee X  = \bigsqcup_\alpha X$.
\end{proposition}
\begin{proof}
First recall that since $x_i = \bigsqcup_\alpha \{x_i\}$, $x_i$ is a
$\sqsubseteq_\alpha$-least and $\leq$-least in $[x_i]_\alpha$.
Let $x = \bigvee X$. Then $x =_\beta x_i$ for all $i \in I$
and $\beta < \alpha$, by Axiom 4. Since also $x_i \leq x$ for all $i \in I$,\
and since Axiom 5 holds, we have $x_i \sqsubseteq_\alpha x$ for all $i \in I$.

Suppose now that $x \sqsubseteq_\alpha z$. We want to prove that $x\leq z$.
Since $x \sqsubseteq_\alpha z$, it holds $x_i \sqsubseteq_\alpha z$ for all $i \in I$.
Since $x_i = \bigsqcup_\alpha \{x_i\}$ for all $i \in I$, it follows that
$x_i \leq z$ for all $i \in I$. Thus, $x = \bigvee X \leq z$, proving that
$\bigvee X \in L_\alpha$.

It remains to show that $x =\bigsqcup_\alpha X$. We have seen that
$x_i \sqsubseteq_\alpha x$ for all $i \in I$. Thus,
$\bigsqcup_\alpha X \leq x$. Also, $x_i \sqsubseteq_\alpha \bigsqcup_\alpha X$
for all $i \in I$, and since by Lemma~\ref{lem-slice added}
 $\bigsqcup_\alpha X\in L_\alpha$,
it follows from Proposition~\ref{prop-ax5-2} that
$x_i \leq \bigsqcup_\alpha X$ for all $i \in I$. Thus,
$x \leq \bigsqcup_\alpha X$.
\end{proof}

\section{The Fixed Point Theorem}\label{fixed-point-section}
In this section we develop a fixed point theorem for functions $f:L\rightarrow L$,
where $L$ is a model of Axioms 1-4. Notice that the functions $f$ we consider are not
necessarily monotonic with respect to $\sqsubseteq$ (and therefore the traditional theorems
of fixed point theory do not apply to them). Instead, we require that the functions we consider
are {\em $\alpha$-monotonic} with respect to $\sqsubseteq_\alpha$, for all $\alpha < \lord$:

\begin{definition}
Suppose that $L$ is a model and let $\alpha < \kappa$.
A function $f: L \to L$ is called $\alpha$-monotonic
if for all $x,y \in L$, if $x \sqsubseteq_\alpha y$ then
$f(x) \sqsubseteq_\alpha f(y)$.
\end{definition}

In order to prove the main theorem (Theorem~\ref{main-fixed-point-theorem}) we need the following
technical lemma.
\begin{lemma}
\label{lem-fix2}
Let $L$ be a model. Suppose that $f: L \to L$ is $\alpha$-monotonic,
where $\alpha < \lord$. If $x \in L$ and  $x \sqsubseteq_\alpha f(x)$,
then there is some $y \in L$ with the following properties:
\begin{itemize}
\item $x \sqsubseteq_\alpha y =_\alpha f(y)$.
\item If $x \sqsubseteq_\alpha z$ and $f(z) \sqsubseteq_\alpha z$, then
$y \sqsubseteq_\alpha z$.
\item $y$ is the $\leq$-least element and a $\sqsubseteq_{\alpha + 1}$-least
element of $[y]_\alpha$ and $y \sqsubseteq_{\alpha + 1} f(y)$.
\end{itemize}
\end{lemma}

\begin{proof}
%
Our proof is similar to a well-known proof of Tarski's least fixed point theorem
that constructs the least fixed point of a monotonic endofunction of a complete lattice
by a transfinite sequence of approximations.
Define $x_0= x$ and for all ordinals $\gamma > 0$ let
$x_\gamma = f(x_\delta)$ if $\gamma = \delta + 1$ and
$x_\gamma = \bigsqcup_\alpha \{x_\delta : \delta < \gamma\}$ if $\gamma > 0$
is a limit ordinal.
The definition makes sense since
we can prove by induction on $\gamma$ that
$x \sqsubseteq_\alpha x_\gamma$ for all $\gamma$.
Indeed, this is clear when $\gamma$ is $0$. Suppose that
$\gamma > 0$ and our claim holds for all ordinals less than $\gamma$.
If $\gamma = \delta+1$, then $x \sqsubseteq_\alpha x_\delta$
by the induction hypothesis, thus $x \sqsubseteq_\alpha f(x) \sqsubseteq_\alpha
f(x_\delta) = x_\gamma$ by the assumption that $f$ is $\alpha$-monotonic
and $x \sqsubseteq_\alpha f(x)$.
 If $\alpha$ is a limit ordinal,
then $x_\gamma = \bigsqcup_\alpha \{x_\delta : \delta < \gamma\}$,
and since $x \sqsubseteq_\alpha x_\delta$ for all
$\delta < \gamma$, also $x \sqsubseteq_\alpha x_\gamma$ by the definition of
$\bigsqcup_\alpha$.

{\it Claim:} For all $\gamma$, $x_\gamma \sqsubseteq_\alpha f(x_\gamma)$.

We prove this claim by induction on $\gamma$.
When $\gamma = 0$, this holds by assumption. Suppose that
$\gamma = \delta + 1$. By the induction hypothesis, we
have $x_\delta \sqsubseteq_\alpha f(x_\delta) = x_\gamma$.
Thus, $x_\gamma = f(x_\delta) \sqsubseteq_\alpha f(x_\gamma)$,
since $f$ is $\alpha$-monotonic.
Suppose now that $\gamma > 0$ is a limit ordinal, so that
$x_\gamma = \bigsqcup_\alpha \{x_\delta : \delta < \gamma\}$.
Thus, $x_\delta \sqsubseteq_\alpha x_\gamma$, and
$x_\delta  \sqsubseteq_\alpha f(x_\delta) \sqsubseteq_\alpha f(x_\gamma)$
for all $\delta < \gamma$,
by the induction hypothesis and since $f$
preserves the relation $\sqsubseteq_\alpha$.
 It follows that $x_\gamma \sqsubseteq_\alpha f(x_\gamma)$.

{\it Claim:} For all ordinals $\beta < \gamma$, we have $x_\beta \sqsubseteq_\alpha x_\gamma$.

Again, we prove this claim by induction on $\gamma$. When $\gamma = 0$, our claim is trivial.
Suppose that $\gamma > 0$. If $\gamma = \delta + 1$, then $\beta \leq \delta$ and
$x_\beta \sqsubseteq_\alpha x_\delta \sqsubseteq_\alpha f(x_\delta) = x_\gamma$
by the induction hypothesis and the previous claim.
Thus, $x_\beta \sqsubseteq_\alpha x_\gamma$.
If $\gamma > 0$ is a limit ordinal, then $x_\beta \sqsubseteq_\alpha x_\gamma$
by the definition of $x_\gamma$.

{\it Claim:} Suppose that $x \sqsubseteq_\alpha z$ and $f(z) \sqsubseteq_\alpha z$.
Then $x_\gamma \sqsubseteq_\alpha z$ for all $\gamma$.

We proceed by induction on $\gamma$. Since $x_0 = x$, our claim is clear
for $\gamma = 0$. Suppose that $\gamma > 0$. If $\gamma = \delta + 1$,
then since $x_\delta \sqsubseteq_\alpha z$ by the induction hypothesis,
we have $x_\gamma = f(x_\delta) \sqsubseteq_\alpha f(z) \sqsubseteq_\alpha z$
by assumption and since $f$ is $\alpha$-monotonic. Thus, $x_\gamma \sqsubseteq_\alpha z$.
If $\gamma > 0$ is a limit ordinal, then $x_\gamma =\bigsqcup_\alpha \{x_\delta : \delta < \gamma\}
\sqsubseteq _\alpha z$, since by the induction hypothesis,
$x_\delta \sqsubseteq_\alpha z$ for all $\delta < \gamma$.

Now, there is an ordinal $\lambda_0$ such that $x_\gamma=_\alpha x_\delta$ for all
$\gamma,\delta \geq \lambda_0$ (since otherwise the cardinality of the set $\{x_\gamma :
\gamma \mbox{ is an ordinal}\}$ would exceed the cardinality of $L$). Let $\lambda$ denote
the least limit ordinal with $\lambda \geq \lambda_0$.
Let $y = x_\lambda$. By the definition of $y$, we have that $f(y) =_\alpha y$ and $x\sqsubseteq_\alpha y$ (since $x = x_0$).
Suppose that $z\in L$ with $x \sqsubseteq_\alpha z$ and
$f(z) \sqsubseteq_\alpha z$. Then, as shown above,
$x_\gamma \sqsubseteq_\alpha z$ for all $\gamma$,
thus $y \sqsubseteq_\alpha z$.

To complete the proof, we still need to verify that $y$ is the
$\leq$-least element of $[y]_\alpha$ and $y  \sqsubseteq_{\alpha + 1} f(y)$.
But $y = \bigsqcup_\alpha \{x_\gamma : \gamma <\lambda\}$,
and thus, by Lemma~\ref{lem-lem0}, $y$ is the $\leq$-least element and a $\sqsubseteq_{\alpha + 1}$-least element
of $[y]_\alpha$. Now $f(y) =_\alpha y$, so $f(y) \in [y]_\alpha$, thus
$y \sqsubseteq_{\alpha + 1} f(y)$.
\end{proof}

Below for any $x\in L$
and ordinal $\alpha < \lord$ with $x \sqsubseteq_\alpha f(x)$,
we will denote the element $y$ constructed
above by $f_\alpha(x)$.  We have shown above that when
$x \sqsubseteq_\alpha f(x)$, then $f_\alpha(x)$
satisfies the three properties of
Lemma~\ref{lem-fix2}.

We now introduce the notion of {\em $\alpha$-continuity}, which
is stronger than $\alpha$-monotonicity:
\begin{definition}
Suppose that $L$ is a model and let $\alpha < \kappa$.
The function $f:L \to L$ is called \emph{$\alpha$-continuous} if it is $\alpha$-monotonic
and for all sequences $(x_n)_{n \geq 0}$ of elements of $L$ such that for all $n\geq 0$,
$x_n \sqsubseteq_\alpha x_{n+1}$, it holds that
$f(\bigsqcup_\alpha \{x_n:n \geq 0\}) =_\alpha \bigsqcup_\alpha \{f(x_n):n \geq 0\}$.
\end{definition}
%
%
%

As the following example illustrates, not every $\alpha$-monotonic function
is $\alpha$-continuous.
\begin{example}
We construct a function over the standard model (ie., the set of infinite-valued
interpretations) which is 0-monotonic but not 0-continuous. Let $Z=\{x_0,x_1,\ldots\}$
be a set of propositional atoms. Consider the following function $f:V^Z\rightarrow V^Z$:
$$f(I)(x) = \mbox{if $\,\,(\forall y\in Z(I(y)=T_0))$ then $T_0$ else $F_0$}$$
Let us denote by $\perp$ (respectively $\top$) the interpretation which assigns to all $x\in Z$ the
value $F_0$ (respectively $T_0$). Then, $f(\top)=\top$ and $f(I) = \perp$, for all $I\neq \top$.
Using this remark, it is clear that $f$ is 0-monotonic. However, $f$ is not 0-continuous. To see this,
consider the chain $I_0\sqsubseteq_0 I_1 \sqsubseteq_0\cdots$ where each $I_n$ is defined as follows:
\[
\begin{array}{ll}
I_n(x_m) = \left\{\begin{array}{ll}
                                    T_0,       &\mbox{if $m<n$}\\
                                    F_0,       &\mbox{otherwise}
                                    \end{array}
                                  \right.
\end{array}
\]
Then, $f(\bigsqcup_\alpha\{I_n:n\geq 0\}) = \top$ while $\bigsqcup_\alpha\{f(I_n):n\geq 0\} = \perp$.
\end{example}

The additional assumption of $\alpha$-continuity is quite important since it reduces in the general case the steps required
in order to obtain the element $f_\alpha(x)$ in the proof of Lemma~\ref{lem-fix2}.
\begin{remark}\label{reducing-steps}
Let $f:L \rightarrow L$ be $\alpha$-continuous for each ordinal $\alpha < \lord$.
Then for each $\alpha$, the construction of the element
$f_\alpha(x)$ in the proof of Lemma~\ref{lem-fix2} terminates at stage $\omega$, since
$f(x_\omega) = f(\bigsqcup_\alpha \{x_n : n \geq 0\}) =
f(\bigsqcup_\alpha \{f^n(x) : n \geq 0\}) =_\alpha  \bigsqcup_\alpha \{f^n(x) : n \geq 0\}
= \bigsqcup_\alpha \{x_n : n \geq 0\} = x_\omega$.
\end{remark}

Using Lemma~\ref{lem-fix2} we can now obtain the main theorem of the paper:
\begin{theorem}\label{main-fixed-point-theorem}
Let $L$ be a model. Suppose that $f: L \to L$ is $\alpha$-monotonic
for each ordinal $\alpha < \lord$. Then $f$ has a least pre-fixed point
with respect to the partial order $\sqsubseteq$, which is also the least
fixed point of $f$.
\end{theorem}

\begin{proof}
Let us define for each ordinal $\alpha < \lord$,
$x_\alpha = f_\alpha(y_\alpha)$, where $y_\alpha = \bigvee_{\beta < \alpha} x_\beta$.
Notice that $x_0 = f_0(\perp)$ and when $\alpha = \beta + 1$, then $x_\alpha = f_\alpha(x_\beta)$.

We need to verify that $y_\alpha \sqsubseteq_{\alpha} f(y_\alpha)$ for all $\alpha < \lord$,
so that all the consequences of Lemma~\ref{lem-fix2} also hold.
We will also show that if $\beta < \alpha$, then $x_\beta =_\beta x_\alpha$.
It then follows that $(x_\alpha)_{\alpha < \lord}$ is a compatible sequence (since
by Lemma~\ref{lem-fix2} each $x_\alpha$ is $\leq$-least in $[x_\alpha]_\alpha$),
in particular $x_\beta \leq x_\alpha$ whenever $\beta < \alpha$.

When $\alpha = 0$, the above facts are clear, since by Remark~\ref{rem-remark1}
it holds $\perp \sqsubseteq_0 z$ for all $z$. Suppose that $\alpha > 0$ and that we have
proved our claim for all ordinals less than $\alpha$. If $\alpha = \beta+1$, then
$y_\alpha = x_\beta$. Since by the induction hypothesis $y_\beta \sqsubseteq_\beta f(y_\beta)$,
we have $x_\beta \sqsubseteq_\alpha f(x_\beta)$ by Lemma~\ref{lem-fix2}.
Thus, $y_\alpha \sqsubseteq_\alpha f(y_\alpha)$.
Also, if $\gamma < \alpha$ then $\gamma \leq \beta$,
so that $x_\gamma =_\gamma x_\beta$
by the induction hypothesis. Since $x_\beta \sqsubseteq_\alpha f_\alpha(x_\beta) =x_\alpha$,
we conclude by Lemma~\ref{some-consequences}(d) that $x_\gamma =_\gamma x_\alpha$.

Suppose now that $\alpha > 0$ is a limit ordinal.
Then $y_\alpha = \bigvee_{\beta < \alpha}x_\beta$,
and since $x_\beta =_\beta x_\gamma$ and $x_\beta \leq x_\gamma$
for all $\beta < \gamma <\alpha$, we know
by Lemma~\ref{lem-sup}
 that $y_\alpha =_\beta x_\beta$ for all
$\beta < \alpha$. Moreover, by Lemma~\ref{lem-partial} we know that $y_\alpha$
is the $\leq$-least and an $\sqsubseteq_{\alpha}$-least
element of $(y_\alpha]_\alpha$. Now since $y_\alpha =_\beta x_\beta$
for all $\beta < \alpha$, also $f(y_\alpha) =_\beta f(x_\beta) =_\beta x_\beta$
for all $\beta < \alpha$. This means that $f(y_\alpha) \in (y_\alpha]_\alpha$
and thus $y_\alpha \sqsubseteq_\alpha f(y_\alpha)$.

Thus, $x_\alpha$ has the properties implied by Lemma~\ref{lem-fix2}.
In particular, $y_\alpha \sqsubseteq_\alpha x_\alpha$ and thus
$x_\beta =_\beta y_\alpha =_\beta x_\alpha$ whenever $\beta < \alpha$
and $x_\alpha$ is the $\leq$-least element of $[x_\alpha]_\alpha$.

We have thus proved that $(x_\alpha)_{\alpha < \lord}$ is a compatible
sequence which determines the element $x_\infty = \bigvee_{\alpha < \lord} x_\alpha$
which is the unique element with $x_\infty =_\alpha x_\alpha$
for all $\alpha < \lord$.

Now since $x_\infty =_\alpha x_\alpha$ for all $\alpha < \lord$,
also $f(x_\infty) =_\alpha f(x_\alpha) =_\alpha x_\alpha$
for all $\alpha < \lord$. Thus, $f(x_\infty) = x_\infty$.
It remains to show that $x_\infty$ is the least pre-fixed
point of $f$ with respect to $\sqsubseteq$.

Suppose that $f(z) \sqsubseteq z$. We want to prove by induction that
for all $\alpha < \lord$, either $x_\gamma \sqsubset_\gamma z$
for some $\gamma < \alpha$, or $x_\gamma \sqsubseteq_\gamma z$
for all $\gamma \leq \alpha$. It then follows that $x_\infty \sqsubseteq z$.

When $\alpha = 0$ then by Lemma~\ref{inclusion-lemma} and Remark~\ref{rem-remark1} it holds
that $f(z) \sqsubseteq_0 z$ and $\perp \sqsubseteq_0 z$;
thus, by Lemma~\ref{lem-fix2}, $x_0 \sqsubseteq_0 z$.

Suppose now that $\alpha > 0$. If $x_\gamma \sqsubset_\gamma z$
for some $\gamma < \alpha$, then we are done. So suppose that
this is not the case, ie., $x_\gamma =_\gamma z$
for all $\gamma < \alpha$.

Suppose that $\alpha = \beta + 1$. Then $x_\beta =_\beta z$ and thus
$z \in [x_\beta]_\beta$. Since $x_\beta$ is $\leq$-least in $[x_\beta]_\beta$,
by Lemma~\ref{lem-lem0} $x_\beta$ is $\sqsubseteq_\alpha$-least
in $[x_\beta]_\beta$, thus $x_\beta \sqsubseteq_\alpha z$.
We conclude that $x_\beta \sqsubseteq_\alpha f(x_\beta) \sqsubseteq_\alpha f(z)$,
ie., $x_\beta \sqsubseteq_\alpha f(z)$.
If $f(z) \sqsubset_\gamma z$ for some $\gamma < \alpha$, then by
$x_\beta \sqsubseteq_\alpha f(z) \sqsubset_\gamma z$ we
have $x_\beta \sqsubset_\gamma z$, contradicting $x_\beta =_\beta z$.
Thus, since $f(z) \sqsubseteq z$ and $f(z) =_\gamma z$ for all $\gamma < \alpha$,
we must have $f(z) \sqsubseteq_\alpha z$. Since $x_\alpha = f_\alpha(x_\beta)$
and $x_\beta \sqsubseteq_\alpha z$, we conclude
by Lemma~\ref{lem-fix2} that $x_\alpha \sqsubseteq_\alpha z$.

Suppose that $\alpha > 0$ is a limit ordinal. Then,
as shown above,  $y_\alpha =_\gamma x_\gamma$
for all $\gamma < \alpha$. Since also $x_\gamma =_\gamma z$
for all $\gamma < \alpha$, we have  $z \in (y_\alpha]_\alpha$.
But $y_\alpha$ is the $\leq$-least and a $\sqsubseteq_\alpha$-least
element of $(y_\alpha]_\alpha$, so that $y_\alpha \sqsubseteq_\alpha z$
and thus $y_\alpha  \sqsubseteq_\alpha f(y_\alpha) \sqsubseteq_\alpha f(z)$;
therefore $y_\alpha \sqsubseteq_\alpha f(z)$. Suppose that $f(z) \sqsubset_\gamma z$
for some $\gamma < \alpha$. Then $y_\alpha \sqsubseteq_\alpha f(z) \sqsubset_\gamma z$
and thus $y_\alpha \sqsubset_\gamma z$, contradicting $y_\alpha \sqsubseteq_\alpha z$.
Thus, $f(z) =_\gamma z$ for all $\gamma < \alpha$.
Since $f(z) \sqsubseteq z$ and $f(z)=_\gamma z$
for all $\gamma < \alpha$,  we have $f(z) \sqsubseteq_\alpha z$.
Since $y_\alpha \sqsubseteq_\alpha z$ and $f(z) \sqsubseteq_\alpha z$,
by Lemma~\ref{lem-fix2} we have
that $x_\alpha \sqsubseteq_\alpha z$.
\end{proof}

The above theorem has as a special case the well-known Knaster-Tarski fixed point theorem~\cite{Tar55}.
This can be seen as follows. As remarked at the end of Subsection~\ref{axioms-subsection},
Theorem~\ref{main-fixed-point-theorem} continues to hold even if $\lord$ is assumed
to be a successor ordinal. Consider now the case $\lord = 2$ and take $\sqsubseteq_0$ to be
equal to the $\leq$ relation; notice that $\sqsubseteq_1$ is the equality relation on $L$.
Then, the statement of the theorem reduces to the Knaster-Tarski theorem when $f$ is assumed
to be monotonic with respect to $\leq$. Moreover, if $f$ is continuous with respect to $\leq$,
then by Remark~\ref{reducing-steps}, Theorem~\ref{main-fixed-point-theorem} reduces to Kleene's
fixed point theorem.

\section{A Class of $\alpha$-Continuous Functions}\label{a-continuous-functions}
In this section we investigate conditions which guarantee that a function
is $\alpha$-monotonic (respectively $\alpha$-continuous). In our exposition
we will need a slight generalization of the definitions of $\alpha$-monotonicity and
$\alpha$-continuity in order to cover functions of the form $f:L\rightarrow L'$
(and not just $f:L\rightarrow L$):
\begin{definition}
Suppose that $L$ and $L'$ are models and $\alpha < \kappa$.
A function $f: L \to L'$ is called $\alpha$-monotonic
if for all $x,y \in L$, if $x \sqsubseteq_\alpha y$ then
$f(x) \sqsubseteq_\alpha f(y)$.
\end{definition}
\begin{definition}
Suppose that $L$ and $L'$ are models and $\alpha < \kappa$.
A function $f:L \to L'$ is called \emph{$\alpha$-continuous} if it is
$\alpha$-monotonic and for all
sequences $(x_n)_{n \geq 0}$ of elements of $L$ such that for all $n\geq 0$,
$x_n \sqsubseteq_\alpha x_{n+1}$, it holds that
$f(\bigsqcup_\alpha \{x_n:n \geq 0\}) =_\alpha \bigsqcup_\alpha \{f(x_n):n \geq 0\}$.
\end{definition}

The following lemma characterizes certain properties of $\alpha$-monotonic and
$\alpha$-continuous functions:
\begin{lemma}\label{composition-projection-lemma}
Suppose that $L,L',L''$ and $L_i, \ i \in I$, are models.
\begin{itemize}
\item
If $f: L \to L'$ and $g: L'\to L''$ are $\alpha$-monotonic ($\alpha$-continuous),
then so is their composition $g\circ f: L \to L''$.
\item
Each projection function $\mathit{pr}_j:\prod_{i \in I} L_i \to L_j$
for $j \in I$ is $\alpha$-continuous and $\alpha$-monotonic.
\item A function $f :L \to \prod_{i\in I}L_i$ is $\alpha$-monotonic
($\alpha$-continuous) iff each
function $\mathit{pr}_i  \circ f: L \to L_i$
is.
\end{itemize}
\end{lemma}

\begin{proof}
To prove the first claim, suppose that $x \sqsubseteq_\alpha y$ in $L$.
If $f,g$ are $\alpha$-monotonic, then $f(x) \sqsubseteq_\alpha f(y)$ and
$g(f(x))\sqsubseteq_\alpha g(f(y))$. Suppose that $(x_n)_{n\geq 0}$ is an $\omega$-chain
in $L$ with $x_n \sqsubseteq_\alpha x_{n+1}$ for all $n \geq 0$. If
$f$ and $g$ are $\alpha$-continuous they are also $\alpha$-monotonic and therefore
$(f(x_n))_{n\geq 0}$ is an $\omega$-chain in $L'$ with $f(x_n) \sqsubseteq_\alpha f(x_{n+1})$
for all $n \geq 0$. Moreover:
\begin{eqnarray*}
g(f(\bigsqcup_\alpha \{x_n:n\geq 0\})) &=_\alpha& g(\bigsqcup_\alpha \{f(x_n):n\geq 0\})\\
&=_\alpha& \bigsqcup_\alpha \{g(f(x_n)):n\geq 0\}.
\end{eqnarray*}

To prove the second claim, suppose that $(x_n)_{n\geq 0}$ is an $\omega$-chain in
$\prod_{i \in I} L_i$ with $x_n \sqsubseteq_\alpha x_{n+1}$ for all $n \geq 0$.
Then for each $j \in I$, $(x_n(j))_{n\geq 0}$ is an $\omega$-chain in $L_j$ with
$x_n(j) \sqsubseteq_\alpha x_{n+1}(j)$ for all $n \geq 0$ and
$(\bigsqcup_\alpha \{x_n:n\geq 0\})(j) = \bigsqcup_\alpha \{x_n(j):n\geq 0\}$. Thus,
\begin{eqnarray*}
\pr_j(\bigsqcup_\alpha \{x_n:n\geq 0\}) &=& (\bigsqcup_\alpha \{x_n:n\geq 0\})(j)\\
&=& \bigsqcup_\alpha \{x_n(j):n\geq 0\}\\
&=& \bigsqcup_\alpha \{\pr_j(x_n):n\geq 0\}.
\end{eqnarray*}

To prove the last claim, suppose that $f: L \to \prod_{i\in I} L_i$. If $f$
is $\alpha$-monotonic ($\alpha$-continuous),
then so is each $\pr_i \circ f$ for $i \in I$ by the first two claims. Suppose now that
each $\pr_i \circ f$ for $i \in I$ is $\alpha$-monotonic and $x \sqsubseteq_\alpha y$.
Then $(f(x))(i) \sqsubseteq_\alpha (f(y))(i)$ holds for all $i \in I$
and thus $f(x) \sqsubseteq_\alpha f(y)$. Suppose now that each $\pr_i \circ f$ is
$\alpha$-continuous for $i \in I$ and let $(x_n)_{n\geq 0}$ be an $\omega$-chain in $L$
with $x_n \sqsubseteq_\alpha x_{n+1}$ for all $n\geq 0$.
Then for each $i\in I$,  $(f(x_n)(i))_{n\geq 0}$ is an $\omega$-chain in $L_i$
with $f(x_n)(i) \sqsubseteq_\alpha f(x_{n+1})(i)$ for all $n \geq 0$
and $f(\bigsqcup_\alpha \{x_n:n\geq 0\})(i) =_\alpha \bigsqcup_\alpha \{f(x_n)(i):n\geq 0\}$.
Thus, $f(\bigsqcup_\alpha \{x_n:n\geq 0\}) =_\alpha \bigsqcup_\alpha \{f(x_n):n\geq 0\}$, proving that
$f$ is $\alpha$-continuous.
\end{proof}

We now proceed to investigate conditions that guarantee that functions are $\alpha$-monotonic
($\alpha$-continuous). As a first step we will impose two new axioms
on our models.

\begin{framed}
\begin{axiom}\label{axiom6}
Suppose that $L$ is a model. We say that $L$ satisfies Axiom 6
if for every $\alpha < \lord$, for every index set $J$ and for all $x_j,y_j \in L$ with $j \in J$,
if $x_j \sqsubseteq_\alpha y_j$ for all $j\in J$, then $\bigvee_{j\in J}x_j \sqsubseteq_\alpha \bigvee_{j \in J} y_j$.
\end{axiom}
\end{framed}

Note that Axiom 6 is obvious when $J$ is empty or a singleton set.
Clearly, Axiom 6 implies Axiom 4. Indeed, if Axiom 6 holds and
$x_j =_\alpha y$ for all $j \in J$, where $J$ is not empty,
then $\bigvee_{j\in J} x_j =_\alpha \bigvee_{j\in J}y = y$.

We have the following lemma and proposition:

 \begin{lemma}\label{product-axiom-6}
 Suppose that $L_i,\ i\in I$, is a family of models satisfying Axiom 6.
 Then $L = \prod_{i\in I} L_i$ also satisfies Axiom 6.
 \end{lemma}

 \begin{proof}
 Let $x_j,y_j\in L$ for all $j \in J$. Suppose that $x_j \sqsubseteq_\alpha y_j$
 for all $j \in J$. Then $x_j(i) \sqsubseteq_\alpha y_j(i)$ for all $i\in I$ and $j \in J$.
 By assumption, Axiom 6 holds in each $L_i$. Thus, we have $\bigvee_{j \in J} x_j(i) \sqsubseteq_\alpha
 \bigvee_{j \in J} y_j(i)$
 for each $i\in I$. We conclude that $\bigvee_{j \in J} x_j \sqsubseteq_\alpha \bigvee_{j \in J}y_j$.
 \end{proof}

 \begin{proposition}
 Suppose that $L,L'$ are models such that $L'$ satisfies Axiom 6 and $\alpha < \kappa$.
 If $f_j: L\to L'$
 is an  $\alpha$-monotonic function for  each $j \in J$, then so is
 $f = \bigvee_{j\in J}f_j: L \to L'$ defined by $f(x) = \bigvee_{j \in J}f_j(x)$.
 \end{proposition}

 \begin{proof}
 Suppose that $x \sqsubseteq_\alpha y$ in $L$. Then $f_j(x) \sqsubseteq_\alpha
 f_j(y)$ for all $j \in J$. Thus, by Axiom 6, $f(x) \sqsubseteq_\alpha f(y)$.
 \end{proof}

\begin{framed}
\begin{axiom}\label{axiom7}
Suppose that $L$ is a model. We say that $L$ satisfies Axiom 7
if it satisfies Axiom 6 and for every $\alpha < \lord$, for every
index set $J$ and for all  $x_{j,n}\in L$ with $j \in J$ and $n \geq 0$,
if $x_{j,n}\sqsubseteq_\alpha x_{j,n+1}$
for all $j \in J$ and $n \geq 0$, then:
\begin{eqnarray*}\label{eq-ax7-1}
 \bigsqcup_\alpha \{\bigvee_{j \in J} x_{j,n}:n\geq 0\}
&=_\alpha& \bigvee_{j \in J} \bigsqcup_{\alpha} \{x_{j,n}:n\geq 0\}.
\end{eqnarray*}
\end{axiom}
\end{framed}

\begin{lemma}\label{product-axiom-7}
Suppose that $L_i,\ i\in I$ is a family of models satisfying Axiom 7.
Then $L = \prod_{i\in I} L_i$ also satisfies Axiom 7.
\end{lemma}

\begin{proof}
Suppose that $L_i,\ i\in I$ is a family of models satisfying Axiom 7.
By Lemma~\ref{product-axiom-6} we know that $L = \prod_{i\in I} L_i$
satisfies Axiom 6.
Let $x_{j,n}\in L$ for all $j \in J$ and $n\geq 0$, with
$x_{j,n}\sqsubseteq_\alpha x_{j,n+1}$.
Then $x_{j,n}(i) \sqsubseteq_\alpha x_{j,n}(i)$ for all $j \in J, n \geq 0$
and $i \in I$. Since by our assumption each $L_i$ satisfies Axiom 7,
we have
\begin{eqnarray*}
\left( \bigsqcup_\alpha \{\bigvee_{j \in J} x_{j,n}:n\geq 0\} \right)(i)
&=&
\bigsqcup_\alpha \{\bigvee_{j \in J} x_{j,n}(i):n\geq 0\}\\
&=_\alpha&
 \bigvee_{j \in J} \bigsqcup_{\alpha}\{ x_{j,n}(i):n\geq 0\}\\
 &=&
\left( \bigvee_{j \in J} \bigsqcup_{\alpha} \{x_{j,n}:n\geq 0\} \right)(i)
\end{eqnarray*}
for each $i \in I$. Thus, $ \bigsqcup_\alpha \{\bigvee_{j \in J} x_{j,n}:n\geq 0\}
=_\alpha \bigvee_{j \in J} \bigsqcup_{\alpha} \{x_{j,n}:n\geq 0\}$, completing the proof of the
fact that $L$ satisfies Axiom 7.
\end{proof}

\begin{proposition}\label{proposition-axiom-7}
Suppose that $L,L'$ are models such that $L'$ satisfies Axiom 7 and $\alpha < \kappa$.
If $f_j: L\to L'$ is an $\alpha$-continuous function for  each $j \in J$, then so is
$f = \bigvee_{j\in J}f_j: L \to L'$.
\end{proposition}

\begin{proof}
Let $(x_n)_{n\geq 0}$ be an $\omega$-chain in $L$ with $x_n \sqsubseteq_\alpha x_{n+1}$
for each $n\geq 0$. Then for each fixed $j$, $(f_j(x_n))_{n\geq 0}$ is an $\omega$-chain
in $L'$ with $f_j(x_n) \sqsubseteq_\alpha f_j(x_{n+1})$ for all $n \geq 0$.
Since Axiom 7 holds in $L'$, the $\omega$-chain $(\bigvee_{j \in J} f_j(x_n))_{n\geq 0}$
satisfies $\bigvee_{j \in J} f_j(x_n) \sqsubseteq_\alpha \bigvee_{j \in J} f_j(x_{n+1})$
for all $n \geq 0$. Moreover:
\begin{eqnarray*}
\bigsqcup_\alpha \{\bigvee_{j \in J} f_j(x_n):n\geq 0\}
&=_\alpha&
 \bigvee_{j \in J} \bigsqcup_\alpha \{f_j(x_n):n\geq 0\}\\
 &=_\alpha&
 \bigvee_{j \in J} f_j(\bigsqcup_\alpha \{x_n:n\geq 0\})
\end{eqnarray*}
using Axiom 7 and the assumption that each $f_j$ is $\alpha$-continuous.
Thus, $f(\bigsqcup_\alpha \{x_n:n\geq 0\}) =_\alpha \bigsqcup_\alpha \{f(x_n):n\geq 0\}$, proving that $f$ is
$\alpha$-continuous.
\end{proof}

Recall now that by Lemma~\ref{product-axiom-7}, if $L$ is a model satisfying Axiom 7, then for
each set $Z$, $L^Z$ is also a model satisfying this axiom. In particular, for each $n \geq 0$,
$L^n$ is a model satisfying this axiom. Then, the following corollary summarizes the results
obtained thus far in this section:
\begin{corollary}\label{a-continuous-corollary}
Let $L$ be a model satisfying Axiom 7 and let $Z$ be a set.
Suppose that $f_j : L^{n_j} \to L$ is an $\alpha$-monotonic
($\alpha$-continuous) function for each $j \in J$.
Consider a function $g: L^Z \to L^Z$ such that each component
function\footnote{The component function $g_z$ of $g$ for $z \in Z$ is the function defined
by $g_z(x) = \pr_z(g(x)) = (g(x))(z)$, for all $x \in L^Z$.} $g_z: L^Z \to L$ of $g$ for $z\in Z$  can be constructed from
the projections $\mathit{pr}_{z'}: L^Z \to L$ for $z'\in Z$, the functions $f_j$, $j \in J$
and the constants in $L$ by function composition and the supremum
operation $\bigvee$. Then $g$ is also $\alpha$-monotonic ($\alpha$-continuous).
\end{corollary}
\begin{proof}
A direct consequence of Lemma~\ref{composition-projection-lemma}, Lemma~\ref{product-axiom-7} and
Proposition~\ref{proposition-axiom-7}.
\end{proof}

We can now obtain an application of the results of this section. In particular, we concentrate
on the model $V$ of truth values. Our aim is to obtain the main result of~\cite{RW05} as a special
case of the general framework that has been developed in this paper.

\begin{lemma}
The model $V$ of truth values satisfies Axioms 6 and 7.
\end{lemma}
\begin{proof}
Let $x_j,y_j \in V$
with $x_j \sqsubseteq_\alpha y_j$ for all $j \in J$. Let $x = \bigvee_{j \in J}x_j$
and $y = \bigvee_{j \in J} y_j$. We want to prove that $x \sqsubseteq_\alpha y$.
This is clear when $J = \emptyset$, so below we assume that $J$ is not empty.
If $x < F_\alpha$ then $x_j = y_j$ for all $j \in J$, so that $x = y$.
Suppose now that $x = T_\beta  \geq T_\alpha$. Then there exists some $j \in J$
with $x_j = y_j = T_\beta$, and  $y_j \leq T_\beta$ for all $j\in J$.
Thus, $y = T_\beta$ and  $x = y$ again. Suppose next that $x = F_\alpha$. Then for all
$j \in J$, $x_j = y_j \leq  F_\alpha$, or $x_j = F_\alpha$ and $\order(y_j) \geq \alpha$.
At any rate, $\order(y)\geq \alpha$ so that $x = F_\alpha \sqsubseteq_\alpha y$.
Last, suppose that $F_\alpha < x <T_\alpha$, ie. $\order(x) > \alpha$.
Then $x_j < T_\alpha$ for all $j \in J$ and thus $y_j \leq T_\alpha$ for all $j\in J$.
Moreover, there is at least one $j$ with $F_\alpha < x_j$ so that also $F_\alpha < y_j$.
It follows that $y = T_\alpha$ or $\order(y) > \alpha$. Thus, $x \sqsubseteq_\alpha y$
again.


Next we prove that $V$ satisfies Axiom 7. To this end, let $x_{j,n}\in V$
with $x_{j,n}\sqsubseteq_\alpha x_{j,n+1}$ for all $j \in J$ and $n\geq 0$,
and we again assume that $J \neq \emptyset$. We already know that
$\bigvee_{j \in J} x_{j,n} \sqsubseteq_\alpha \bigvee_{j \in J}x_{j,n+1}$
holds for all $n \geq 0$. Define
\begin{eqnarray*}
y_j &=&\bigsqcup_\alpha \{x_{j,n}: n \geq 0\},\quad j \in J\\
y  &=&  \bigvee_{j \in J} y_j\\
z_n &=& \bigvee_{j \in J} x_{j,n},\quad n \geq 0\\
z &=& \bigsqcup_\alpha \{z_n: n \geq 0\}.
\end{eqnarray*}
Our aim is to prove that $y = z$.

Suppose first that $\order(z) < \alpha$. Then $z_n = z$ for all $n \geq 0$,
and either $z = F_\beta$ or $z = T_\beta$ for some $\beta < \alpha$.
If $z = F_\beta$ for some $\beta < \alpha$ then $x_{j,n} \leq F_\beta$
for all $j\in J$ and $n \geq 0$, and $z = \bigvee_{j\in J,n\geq 0} x_{j,n} (= F_\beta)$.
Moreover, $y_j = x_{j,n} \leq F_\beta$ for all $j \in J$ and $n \geq 0$,
so that $y = \bigvee_{j \in J} y_j = \bigvee_{j \in J,n\geq 0} x_{j,n} = F_\beta$.
Thus, $y = z$. Suppose now that $z = T_\beta$, where $\beta <\alpha$.
Then $z_n = T_\beta$
for all $n\geq 0$ and thus $z = \bigvee_{j \in J,n \geq 0} x_{j,n} =T_\beta$.
Since $z_n = T_\beta$ for all $n \geq 0$, $x_{j,n} \leq T_\beta$
for all $j\in J$ and $n \geq 0$, and for each $n$ there is some $j_0$
with $x_{j_0,n} = T_\beta$. However, if $x_{j_0,n} = T_{\beta}$ for some $n$,
then $x_{j_0,m} = T_\beta$ for all $m\geq 0$ and $y_{j_0} = T_\beta$.
We observe that $y_j\leq T_\beta$ for all $j \in J$, and there is
at least one $j_0\in J$ with $y_{j_0} = T_\beta$.
Thus, $y = \bigvee_{j \in J} y_j = T_\beta$ and $y = z$.

Suppose next that $\order(z) > \alpha$. In this case $z = F_{\alpha + 1}$.
We have $F_\alpha \leq z_n< T_\alpha$ for all $n\geq 0$, moreover, there exists
some $n$ with $F_\alpha < z_n < T_\alpha$. It follows that
$x_{j,n} <T_\alpha$ for all $j\in J$ and $n \geq 0$, moreover,
there exists some $j$ and $n$ with $F_\alpha < x_{j,n}< T_\alpha$.
We conclude that $y_j \leq F_{\alpha+1}$ for all $j\in J$,
and that there is some $j$ with $y_j = F_{\alpha+1}$.
Thus, $y = F_{\alpha+1} = z$.


Suppose last that $\order(z) = \alpha$, so that $z = F_\alpha$ or $z = T_\alpha$.
The case when $z = F_\alpha$ is similar to the case when $z = F_\beta$ for
some $\beta < \alpha$ and we have that $y = F_\alpha$. So suppose that
$z = T_\alpha$. Then there exists $n_0$ such that $z_{n_0}=T_\alpha$ and $j_0$ such
that $x_{j_0,n_0}=T_\alpha$. This implies that $y_{j_0}=T_\alpha$. Moreover, by the
definition of $z$ and by the fact that $z=T_\alpha$ we have that $\order(z_n)\geq \alpha$
for all $n\geq 0$; therefore, by the definition of $z_n$, $x_{j,n} \leq T_\alpha$, for all
$n\geq 0$ and $j\in J$. It follows that $y_j \leq T_\alpha$. Therefore, since $y_{j_0}=T_\alpha$, we get
that $y=T_\alpha = z$.
\end{proof}

As the following two lemmas demonstrate, the negation operation $\mysim \, : V \rightarrow V$ and the
conjunction operation $\wedge : V\times V \rightarrow V$ are both  $\alpha$-continuous. Recall
the definition of $\mysim$ (Definition~\ref{interpretation}); the definition of $\wedge$
(also implicitly given in Definition~\ref{interpretation}) is as follows: for $x,y\in V$,
$x\wedge y = \mathit{min}\{x,y\}$. For simplicity, in the following we will use $\wedge$
both as an infix as-well-as prefix operation.
\begin{lemma}
The conjunction operation $\wedge: V\times V \rightarrow V$ is $\alpha$-continuous,
for all $\alpha<\Omega$.
\end{lemma}
\begin{proof}
We first show that $\wedge$ is $\alpha$-monotonic. Consider $(x_1,y_1) \sqsubseteq_\alpha (x_2,y_2)$.
It suffices to show that $x_1\wedge y_1 \sqsubseteq_\alpha x_2 \wedge y_2$. We perform
a case analysis on the value of $v = \mathit{min}\{x_1,y_1\}$. If $v < F_\alpha$
or $v > T_\alpha$ then, by the
definition of $\sqsubseteq_\alpha$, $\mathit{min}\{x_1,y_1\} = \mathit{min}\{x_2,y_2\}$ and
therefore $x_1 \wedge y_1 = x_2 \wedge y_2$. If $v = F_\alpha$ then
$F_\alpha \leq \mathit{min}\{x_2,y_2\} \leq T_\alpha$ and therefore
$x_1 \wedge y_1 \sqsubseteq_\alpha x_2 \wedge y_2$. If $v = T_\alpha$ then
$\mathit{min}\{x_2,y_2\}=T_\alpha$ and therefore $x_1 \wedge y_1 = x_2 \wedge y_2$. Finally,
if $F_\alpha < v < T_\alpha$ then $F_\alpha < \mathit{min}\{x_2,y_2\} \leq T_\alpha$
and therefore $x_1 \wedge y_1 \sqsubseteq_\alpha x_2 \wedge y_2$.

It remains to show that $\wedge$ is $\alpha$-continuous. Let $((x_n,y_n))_{n\geq 0}$ be a sequence such that
$(x_n,y_n) \sqsubseteq_\alpha (x_{n+1},y_{n+1})$ for all $n\geq 0$. We show that:
$$\wedge (\bigsqcup_\alpha\{(x_n,y_n): n\geq 0\}) =_\alpha \bigsqcup_\alpha \{x_n \wedge y_n: n\geq 0\}$$
or equivalently that:
$$(\bigsqcup_\alpha\{x_n:n \geq 0\}) \wedge (\bigsqcup_\alpha\{y_n:n \geq 0\}) =_\alpha \bigsqcup_\alpha \{x_n \wedge y_n: n\geq 0\}$$
Notice that the right hand side of the above $\alpha$-equality is well-defined since, by the
$\alpha$-monotonicity of $\wedge$, the sequence $((x_n \wedge y_n))_{n\geq 0}$ is an increasing
chain with respect to $\sqsubseteq_\alpha$.

Let $x = \bigsqcup_\alpha\{x_n:n \geq 0\}$ and $y = \bigsqcup_\alpha\{y_n:n \geq 0\}$.
We proceed by a case analysis on $v = \mathit{min}\{x,y\}$. Assume first that $v<F_\alpha$
and, without loss of generality, assume that $x = v$. By the definition of $\sqsubseteq_\alpha$
we get that for all $n \geq 0$, $x_n = v$ and therefore $x_n \wedge y_n = v$. Consequently,
$\bigsqcup_\alpha \{x_n \wedge y_n: n\geq 0\}= v$. The case $v> T_\alpha$ is similar.
Consider now the case $v = F_\alpha$ and, without loss of generality, assume that $x = F_\alpha$.
This implies that for all $n \geq 0$, $x_n = F_\alpha$.
Consequently, $x_n \wedge y_n = F_\alpha$ and therefore
$\bigsqcup_\alpha \{x_n \wedge y_n: n\geq 0\} = F_\alpha$.
Consider now the case $v = T_\alpha$ and, without loss of generality, assume that $x = T_\alpha$.
Then, $y \geq T_\alpha$.  Since $x = T_\alpha$, there exists some $n_0$ such that for all $n \geq n_0$,
$x_n = T_\alpha$. Moreover, since $y \geq T_\alpha$, there exists some $n_1$ such that for all $n \geq n_1$,
$y_n \geq T_\alpha$. Consequently, for all $n \geq \mathit{max}\{n_0,n_1\}$, $x_n \wedge y_n = T_\alpha$ and
therefore $\bigsqcup_\alpha \{x_n \wedge y_n: n\geq 0\} = T_\alpha$.
Finally, consider the case where $F_\alpha < v <T_\alpha$, and without
loss of generality assume that $x=v$. Since $x = \bigsqcup_\alpha\{x_n:n \geq 0\}$, we
have by the definition of $\bigsqcup_\alpha$ that $v=F_{\alpha+1}$.
Moreover, for all $n\geq 0$, $F_\alpha \leq x_n < T_\alpha$,
and there is some $n_0$ such that for all $n\geq n_0$, $F_\alpha < x_n < T_\alpha$.
Since $v=\mathit{min}\{x,y\}$ we have that $y > F_\alpha$, and therefore there exists
some $n_1$ such that for all $n \geq n_1$, $y_n > F_\alpha$.
Thus, for all $n \geq \mathit{max}\{n_0,n_1\}$ it is $F_\alpha < x_n \wedge y_n < T_\alpha$. Consequently,
$\bigsqcup_\alpha \{x_n \wedge y_n: n\geq 0\} = F_{\alpha+1}$.
\end{proof}
\begin{lemma}\label{mysim-is-a-continuous}
The negation operation $\mysim \, : V \rightarrow V$ is $\alpha$-continuous,
for all $\alpha<\Omega$.
\end{lemma}
\begin{proof}
We recall the definition of the negation function $\mysim \,: V \rightarrow V$ (see Definition~\ref{interpretation}):
\[
             \mysim (v) = \left\{
                             \begin{array}{ll}
                             T_{\alpha + 1} & \mbox{if $v = F_\alpha$}\\
                             F_{\alpha + 1} & \mbox{if $v = T_\alpha$}\\
                             0              & \mbox{if $v = 0$}
                             \end{array}
                      \right.
\]
It is straightforward to show that $\mysim$ is $\alpha$-monotonic. We show that
it is also $\alpha$-continuous. To see this, let $(x_n)_{n\geq 0}$ be an $\omega$-chain
of truth values such that $x_n \sqsubseteq_\alpha x_{n+1}$, for all $n$. We show that:
$$\mysim(\bigsqcup_\alpha\{x_n:n\geq 0\}) =_\alpha \bigsqcup_\alpha\{\mysim x_n : n\geq 0\}$$
We distinguish cases based on the value of $x_0$. If $\order(x_0) < \alpha$ then $x_n = x_0$ for
all $n$; the statement then obviously holds since both of its sides are equal to $\mysim x_0$.
If $\order(x_0) = \alpha$ then we distinguish two subcases: if $x_0 = T_\alpha$ then $x_n = T_\alpha$
for all $n$ and the result holds; if $x_0=F_\alpha$ then either $x_n=F_\alpha$ for all $n$, or the
chain elements become $T_\alpha$ from a point on, or after a point of the chain all elements have
order greater than $\alpha$; in all subcases, the above statement holds. Finally, if $\order(x_0) > \alpha$,
then we again distinguish two subcases: either all elements of the chain have order greater than $\alpha$
or after a point in the chain all elements become equal to $T_\alpha$; in both cases the statement holds.
\end{proof}

The above discussion leads to the following lemma:
\begin{lemma}\label{tp-new-proof}
Let $P$ be a program. Then, for all countable ordinals $\alpha < \Omega$, $T_P$ is $\alpha$-continuous.
\end{lemma}
\begin{proof}
Consider the immediate consequence operator $T_P$ as defined in Section~\ref{infinite-valued-section}.
Then, for every propositional atom $p$ of program $P$, the component function of $T_P$ corresponding to $p$
is constructed as the supremum of $\alpha$-continuous functions (since conjunction and negation are
$\alpha$-continuous and since composition preserves $\alpha$-continuity). The result is therefore a
direct consequence of Corollary~\ref{a-continuous-corollary}.
\end{proof}

As we have seen earlier in the paper, the set of infinite-valued interpretations
of a logic program satisfies the axioms of Subsection~\ref{axioms-subsection}. By the above lemma
and Theorem~\ref{main-fixed-point-theorem} we immediately get that $T_P$ has a least fixed point.
In this way we prove Theorem~\ref{least} (Corollary 7.5, page 460 of~\cite{RW05}) in a much more
structured and less ad-hoc way. Actually, the proof of Lemma~\ref{tp-new-proof} suggests that the
least fixed point result also holds if we generalize the class of logic programs we consider by allowing
the bodies of program rules to be arbitrary formulas involving negation, conjunction and any other
$\alpha$-continuous function. In this way we actually obtain a much
more general result than the one established in~\cite{RW05}.

\section{Related Work}\label{related-work-section}
The results reported in this paper are connected to previous work on the
development of an abstract fixed point theory for non-monotonic operators.
Pioneering in this respect is the work of Fitting~\cite{Fit02} who used the
abstract framework of lattices and operators on lattices in order to characterize
all major semantic approaches of logic programming. Despite its abstract nature,
Fitting's work is centered around the theory of logic programming.

The next step in this line of research is reported in~\cite{DMT00,DMT04} where
the authors proposed an abstract fixed point theory whose purpose is to be more
widely applicable than just in logic programming. In order for this to be achieved,
the work in~\cite{DMT00,DMT04} considers an arbitrary complete lattice $L$ and
also arbitrary (ie. not necessarily monotonic) operators $f:L\rightarrow L$.
Instead of studying $L$ directly, one can study the product lattice $L^2$. The
intuition here is that elements of $L^2$ can be considered as approximations
to the elements of $L$. More specifically, one can study the fixed points of $f$ by investigating
the fixed points of its so-called {\em approximation operators}: roughly speaking,
an approximating operator of $f$ is a function $A_f:L^2\rightarrow L^2$
whose fixed points approximate the fixed points of $f$. One characteristic of the
approach developed in~\cite{DMT00} is that in order to study the fixed points of
the operator $f$, one must first {\em choose} in some way an appropriate
approximating operator for $f$ (out of possibly many available). This last point
leads to a main difference between approximation theory and our work. In our
setting, given a lattice $(L,\sqsubseteq)$ that obeys the axioms of
Subsection~\ref{axioms-subsection} and an operator $f$ that preserves the relations
$\sqsubseteq_\alpha$, Theorem~\ref{main-fixed-point-theorem} {\em guarantees} that $f$
has a least fixed point. A second important difference between the present work and the
one developed in~\cite{DMT00,DMT04} is that we are seeking the unique {\em least}
fixed point of $f$ with respect to the ordering relation $\sqsubseteq$. On the other hand,
the work in~\cite{DMT00,DMT04} focuses attention on fixed points that are {\em minimal}
with respect to the corresponding ordering relation (see for example Section 4 of~\cite{DMT00}
and in particular Proposition 24 of the aforementioned article). It would be interesting,
but certainly non-trivial, to find underlying relationships between the present work
and the one reported in~\cite{DMT00,DMT04}.

A more recent work that is also connected to our approach is reported in~\cite{VGD06}.
In that paper the authors consider the case of {\em product lattices} as-well-as
{\em stratifiable} operators on such lattices. In our terminology, an operator
$f:L\rightarrow L$ is stratifiable iff for all $x,y\in L$ and for all $\alpha<\lord$,
if $x=_\alpha y$ then $f(x)=_\alpha f(y)$. The authors demonstrate~\cite{VGD06}[Theorem 3.5]
that for every stratifiable operator $f$ it holds that every fixed point of $f$ can be
constructed using the fixed points of a family of operators called the {\em components} of $f$
(intuitively, to every sublattice $L_i$ of the product lattice $L$ there corresponds
a subfamily of the components). However, the least fixed point of~\cite{VGD06}[Theorem 3.5]
is with respect to the pointwise partial order that is defined on the product lattice $L$
while in our case $\sqsubseteq$ is not necessarily pointwise. Moreover, our construction
does not only apply to product lattices but to all lattices that satisfy the axioms
of Subsection~\ref{axioms-subsection}. For example, the non-standard product model
of Subsection~\ref{non-standard-product-model} does not fall within the scope of the
results developed in~\cite{VGD06}.

In general, we feel that abstract fixed point theory for non-monotonic functions is
an evolving and fruitful area of research that still has a lot to offer.

\section{Conclusions}\label{conclusions-section}
We have presented a novel fixed point theorem (Theorem~\ref{main-fixed-point-theorem})
for a class of non-monotonic functions. The aforementioned theorem gives a direct and
elegant proof of the least fixed point result that was obtained in~\cite{RW05} for
the case of normal logic programs. Actually, as noted at the end of Section~\ref{a-continuous-functions},
the proof we obtain applies to a significantly broader class of logic programs than
the one considered in~\cite{RW05}. Moreover, we believe that Theorem~\ref{main-fixed-point-theorem}
may have applications in other classes of logic programs. One such case is extensional higher-order
logic programming~\cite{CHRW13}, which enhances classical logic programming with higher-order
predicates. We have recently used the main results of the present paper in order to obtain a minimum
model semantics for extensional higher-order logic programming extended with negation~\cite{CER14}.
Other possible areas of logic programming that can benefit from the proposed fixed point theorem are
disjunctive logic programming with negation~\cite{CPRW07} and logic programming
with preferences~\cite{RT13}.

Apart from logic programming, it would be interesting to investigate other
applications of the derived theorem. One natural candidate is the theory of weighted
automata and weighted languages. Indeed, the behavior of a weighted automaton is given
by a function mapping words into a weight structure which is
often a complete lattice (see~\cite{DKV09} for a comprehensive
treatment of weighted automata).  When the weighted automaton is a
``boolean automaton'', then it becomes natural to use complete lattices
enriched with preorderings $\sqsubseteq_\alpha$ where $\alpha$ ranges over
all ordinals less than a given nonzero ordinal $\kappa$.

\vspace{0.5cm}
\noindent
{\bf Acknowledgments}: We would like to thank Angelos Charalambidis and Christos Nomikos for
their comments on previous versions of this paper. We would also like to thank an anonymous reviewer
for providing insightful comments on our original submission.

\vspace{-0.1cm}

\bibliographystyle{alpha}


\end{document}